\documentclass[a4paper,11pt]{article}

\usepackage{amsmath,times,fullpage}
\usepackage{amssymb}
\usepackage[bookmarks=true,
            bookmarksnumbered=true,
            pagebackref=true,
            colorlinks=true,
            linkcolor=blue,
            citecolor=red]{hyperref}
\usepackage{amsthm}
\usepackage{amsfonts}
\usepackage{verbatim}
\usepackage{comment}
\usepackage{times}
\usepackage{fullpage}
\usepackage{color}
\usepackage{graphicx}
\usepackage[all]{xy}
\usepackage{relsize}

\newtheorem{lemma}{Lemma}
\newtheorem{theorem}{Theorem}
\newtheorem{definition}{Definition}
\newtheorem{corollary}{Corollary}
\newtheorem{claim}{Claim}

\DeclareMathOperator{\RRES}{\mathsf{Reg-Res}}
\DeclareMathOperator{\Poly}{\mathsf{Poly}}
\DeclareMathOperator{\pPoly}{\mathsf{paraPoly}}
\DeclareMathOperator{\CC}{\mathsf{C_n}}
\DeclareMathOperator{\BinC}{\mathsf{Bin-Clique_k^n}}

\DeclareMathOperator{\RLOP}{\mathsf{rLOP_n}}
\DeclareMathOperator{\bRLOP}{\mathsf{Bin-rLOP_n}}

\DeclareMathOperator{\LOP}{\mathsf{LOP_n}}
\DeclareMathOperator{\UC}{\mathsf{Un-C_n}}
\DeclareMathOperator{\UCk}{\mathsf{Un-C^k_n}}
\DeclareMathOperator{\BC}{\mathsf{Bin-C_n}}
\DeclareMathOperator{\BCk}{\mathsf{Bin-C^k_n}}
\DeclareMathOperator{\bLOP}{\mathsf{Bin-LOP_n}}
\DeclareMathOperator{\OP}{\mathsf{OP_n}}
\DeclareMathOperator{\bOP}{\mathsf{Bin-OP_n}}
\DeclareMathOperator{\UnC}{\mathsf{Clique_k^n}}
\DeclareMathOperator{\Cl}{\mathsf{Clique}}
\DeclareMathOperator{\RES}{\mathsf{Res}}
\DeclareMathOperator{\pPHP}{\mathsf{PHP}}
\DeclareMathOperator{\BinPHP}{\mathsf{Bin-PHP}}
\newcommand{\TCPrin}{\ensuremath{\mathsf{TC\mbox{-}Prin}}}
\newcommand{\UnTCPrin}{\ensuremath{\mathsf{Un\mbox{-}TC\mbox{-}Prin}}}
\newcommand{\BinTCPrin}{\ensuremath{\mathsf{Bin\mbox{-}TC\mbox{-}Prin}}}

\newcommand{\ignore}[1]{}

\newcommand{\PHP}{\ensuremath{\mathrm{PHP}}}

\usepackage{authblk}
\title{Resolution and the binary encoding of combinatorial principles }
\author[1]{Stefan Dantchev}
\author[2]{Nicola Galesi}
\author[1]{Barnaby Martin}
\affil[1]{Department of Computer Science, University of Durham}
\affil[2]{Dipartimento di Informatica, Sapienza Universit\`a Roma}

\begin{document}
\maketitle

\begin{abstract}
We investigate the size complexity of proofs in $\RES(s)$ -- an extension of Resolution working on $s$-DNFs
instead of clauses --  for families of contradictions given in the {\em unusual  binary} encoding. 
A motivation of our work is size lower bounds of refutations in Resolution for families of contradictions in the {\em usual unary} encoding. Our main interest is the $k$-Clique Principle, whose  Resolution complexity is still unknown.  
The approach is justified by the observation that for a large class of combinatorial principles (those expressible as $\Pi_2$ first-order formulae)
short $\RES(\log n)$ refutations for the binary encoding are reducible to  
 short  Resolution refutations of  the unary encoding. 

Our main result  is a $n^{\Omega(k)}$ lower bound for the size of refutations of the binary $k$-Clique Principle in $\RES(\lfloor \frac{1}{2}\log \log n\rfloor)$.  This improves the result of Lauria, Pudl\'ak et al. \cite{DBLP:journals/combinatorica/LauriaPRT17} who proved  the lower bound for Resolution, that is $\RES(1)$. A lower bound in $\RES(\log n)$ for the binary $k$-Clique Principle would prove a lower bound in Resolution for its unary version.  Resolution lower bounds for the (unary) $k$-Clique Principle are known only when refutations are either treelike  \cite{Beyersdorff:2013:TOCL} or  read-once  \cite{DBLP:conf/stoc/AtseriasBRLNR18} (regular Resolution).

To contrast the proof complexity between the unary and binary encodings of  combinatorial principles, we consider the binary (weak) Pigeonhole principle $\BinPHP^m_n$ for $m>n$.  
Our second lower bound proves that in $\RES(s)$ for $s\leq \log^{\frac{1}{2-\epsilon}}(n)$, the shortest proofs of the $\BinPHP^m_n$, requires size  $2^{n^{1-\delta}}$, for any $\delta>0$. 

By a result  of Buss and Pitassi \cite{BussP97} we know that for the (unary, weak) Pigeonhole principle $\pPHP^m_n$, exponential lower bounds (in the size of $\pPHP^m_n$) are not possible in Resolution when $m \geq 2^{\sqrt{n\log n}}$ since there is an upper bound of $2^{O(\sqrt{n\log n})}$. Our lower bound for $\BinPHP^m_n$, together with the fact short $\RES(1)$ refutations for $\pPHP^m_n$ can be translated into short $\RES(\log n)$ proofs for $\BinPHP^m_n$, shows a form of tightness of the upper bound of \cite{BussP97}. Furthermore we prove that $\BinPHP^m_n$  can be refuted in size $2^{\Theta(n)}$  in treelike $\RES(1)$, contrasting with the unary case, where $\pPHP^m_n$ requires  treelike $\RES(1)$ \ refutations of size $2^{\Omega(n \log n)}$ \cite{BeyersdorffGL10, DantchevR01}.

In order to compare the complexity of refuting binary encodings  in Resolution with respect to their unary version, 
we  study under what conditions the complexity of refutations in Resolution will not increase significantly (more than a polynomial factor) when shifting between the unary encoding and the binary encoding.   We show that this is true, from unary to binary, for propositional encodings of principles expressible as a 
$\Pi_2$-formula and involving {\em total variable comparisons}. We then show that this is true, from binary to unary, when one considers the \emph{functional unary encoding}.  In particular, we derive a polynomial upper bound in $\RES(1)$ for  the binary version $\bRLOP$ of a variant of the Linear Ordering  principle, $\RLOP$,  which exponentially separates read-once Resolution from Resolution (see \cite{AJPU}).

Finally we prove that the binary encoding of the general  Ordering principle $\bOP$ -- with no total ordering constraints -- is  polynomially provable in Resolution. These last results can  be interpreted as addressing  the property that shifting to the binary encoding is  preserving the proof hardness of the corresponding unary encodings when working in Resolution.  

\end{abstract}

\section{Introduction}

Various fundamental combinatorial principles used in Proof Complexity may be given in first-order logic as sentences $\varphi$ with no finite models. Riis discusses in \cite{SorenGap} how to generate from $\varphi$ a family of CNFs, the $n$th of which encodes that $\varphi$ has a model of size $n$, which are hence contradictions. Following Riis, it is typical to encode the existence of the witnesses in longhand with a big disjunction, that we designate the \emph{unary encoding}. %where the type of the universally quantified object to which their is an existentially quantified witness is distinct from the type of that witness\COM{\TODO{I do not understand this sentence and "witness" is used twice in few words]}}. For such principles, it is typical to encode the existence of the witness in longhand with a big disjunction, that we designate the \emph{unary encoding}.  
As recently investigated in the  works \cite{DBLP:journals/siamcomp/FilmusLNRT15,DBLP:journals/jacm/BonacinaG15,DBLP:journals/siamcomp/BonacinaGT16,DBLP:journals/combinatorica/LauriaPRT17,DBLP:conf/focs/HrubesP17}, it may also be possible to encode the existence of such witnesses {\em succinctly} by the use of a \emph{binary encoding}. Essentially, the existence of the witness is now given implicitly as any propositional assignment to the relevant variables gives a witness, whereas in the unary encoding a solitary true literal tells us which is the witness\footnote{see Subsection \ref{subsec:theory}  in the Introduction for examples and a more formal statement.}.  
Combinatorial principles encoded in binary are interesting to study since, loosely speaking, they still preserve the hardness of the combinatorial  principle encoded while giving a more  succinct propositional representation. In certain cases this leads to obtain significant lower bounds   in an easier way than for the unary case \cite{DBLP:journals/siamcomp/FilmusLNRT15,DBLP:journals/siamcomp/BonacinaGT16,DBLP:journals/combinatorica/LauriaPRT17}. 

The central thrust of this work is to contrast the proof complexity (size) between the unary and binary encodings of natural combinatorial principles. The main motivation is to approach  size lower bounds of refutations in Resolution for families of contradictions in the  usual unary  encoding,  by looking at the complexity of proofs in $\RES(s)$ for the corresponding families of contradictions where witnesses are given in the binary encodings. $\RES(s)$, is a refutational proof system extending Resolution to $s$-bounded DNFs, introduced by Kraj\'{\i}\^{c}ek in \cite{krabook95}.  Our approach is justified by observing that (see Lemma \ref{lem:resloggen}), for a family of contradictions encoding a principle which is expressible as $\Pi_2$ first-order  formulae having no finite models, short $\RES(\log n)$ refutations of  their  {\em binary} encoding can be obtained from short Resolution refutations for the {\em unary} encoding.

Our main interest is the \emph{$k$-Clique Principle}, whose precise Resolution complexity is still unknown; but we also study other principles, to make progress in the direction of our approach.
The three combinatorial principles we deal  with in this paper are: (1) the $k$-Clique Formulas,  $\UnC(G)$;
 (2) the (weak) Pigeonhole Principle $\pPHP^m_n$;   and (3)  the (Linear) Ordering Principle, ($\mathsf{L}$)$\mathsf{OP}_n$.  The \emph{$k$-Clique Formulas} introduced in  \cite{Beyersdorff:2013:TOCL,DBLP:journals/toct/BeyersdorffGLR12,DBLP:conf/coco/BeameIS01} are formulas stating that a given graph $G$ does have a $k$-clique and are therefore unsatisfiable when $G$ does not contain a $k$-clique. 
The  Pigeonhole principle states that a total mapping $f:[m] \rightarrow [n]$ has necessarily a collision when $m>n$. Its propositional formulation in the negation, $\pPHP^m_n$ is well-studied in proof complexity (see among others: \cite{Haken,DBLP:journals/siamcomp/SegerlindBI04,DantchevR01,DBLP:journals/jacm/Raz04,DBLP:journals/tcs/Razborov03,10.1007/3-540-46011-X_8,Ben-sasson99shortproofs,BussP97,BeyersdorffGL10,DBLP:conf/focs/BeameP96,DBLP:journals/iandc/AtseriasBE02,DBLP:journals/tcs/Atserias03,DBLP:journals/jcss/MacielPW02}). 
%\emph{Ordering principles}  state the  property that a finite linearly ordered set have necessarily a minimal element. 
The $\LOP$ formulas encodes the negation of the Linear Ordering Principle which asserts that each finite linearly ordered set has a maximal element and was introduced and studied, among others, in the works \cite{DBLP:journals/acta/Krishnamurthy85,DBLP:journals/acta/Stalmarck96,DBLP:journals/cc/BonetG01}.

\subsection{Contributions}
Deciding whether a graph has a $k$-clique it is one of the central problems in Computer Science and  can be decided in time 
$n^{O(k)}$ by a brute force algorithm. It is then of the utmost  importance to understand whether given algorithmic primitives are sufficient to design algorithms solving the Clique problem more efficiently than the trivial upper bound.  
Resolution refutations for the formula $\UnC(G)$ (respectively any CNF $F$), 
can be thought  as the execution trace of an algorithm, whose primitives are defined by the rules of the Resolution system, 
searching for a $k$-Clique inside $G$ (respectively deciding the satisfiability of $F$).  Hence  understanding whether there are $n^{\Omega(k)}$ size lower bounds in Resolution for refuting $\UnC(G)$ would then 
answer the  above question for algorithms based on Resolution primitives.   This question was posed in \cite{Beyersdorff:2013:TOCL}, 
where it was also answered in the case of refutations in the form of trees (treelike Resolution). 
Recently  in a major breakthrough Atserias et al. in \cite{DBLP:conf/stoc/AtseriasBRLNR18} prove the $n^{\Omega(k)}$ lower bound
for the case of read-once proofs (Regular resolution).  The graph $G$ considered in  \cite{Beyersdorff:2013:TOCL, DBLP:conf/stoc/AtseriasBRLNR18} to plug in the formula  $\UnC(G)$ to make it unsatisfiable was a random graph obtained by  a slight variation of Erd\"os-R\'enyi distribution of random graphs as defined in \cite{Beyersdorff:2013:TOCL}. 
But the exact Resolution complexity of $\UnC(G)$, for $G$ random is unknown. 
In the work \cite{DBLP:journals/combinatorica/LauriaPRT17}, Lauria et al. consider the binary encoding of Ramsey-type propositional statements, having as a special case a binary version of $\UnC(G)$: $\BinC(G)$. They obtain optimal lower bounds for $\BinC(G)$ in Resolution, which is $\RES(1)$. 

Our main result  (Theorem \ref{thm:main-BinC}) is a $n^{\Omega(k)}$ lower bound for the size of refutations of $\BinC(G)$  in $\RES(\frac{1}{2}\log \log n)$, when $G$ is a random  graph as that defined in \cite{Beyersdorff:2013:TOCL}.  Lemma~\ref{lem:reslog} in Section~\ref{sec:k-clique} proves that a lower bound in $\RES(\log)$ for the $\BinC(G)$  would prove a lower bound in Resolution for $\UnC(G)$.

\subsubsection{Weak Pigeonhole principle} 
An interesting example to test the relative hardness of binary versions of combinatorial principle comes from the (weak) Pigeonhole principle. 
In Section \ref{sec:wphp}, we consider its  binary version  $\BinPHP^m_n$  and we prove  that in $\RES(s)$ for $s\leq \log^{\frac{1}{2-\epsilon}}(n)$, the shortest proofs of the $\BinPHP^m_n$, require size  $2^{n^{1-\delta}}$, for any $\delta>0$ (Theorem \ref{thm:wphp}). This is the first size lower bound known for the  $\BinPHP^m_n$ in $\RES(s)$. As a by-product of this lower bound we prove a lower bound of the order $2^{\Omega(\frac{n}{\log n})}$ (Theorem \ref{thm:reswphp}) for the size of the shortest Resolution refutation of $\BinPHP^m_n$.  Our lower bound for $\RES(s)$ is obtained through a technique that merges together, the random restriction method, an inductive argument on the $s$ of $\RES(s)$ and the notion of {\em minimal covering} of a $k$-DNF of \cite{DBLP:journals/siamcomp/SegerlindBI04}. Since we are not using any (even weak) form of Switching Lemma (as for instance in  \cite{DBLP:journals/siamcomp/SegerlindBI04,DBLP:journals/cc/Alekhnovich11}), we consider how tight is our lower bound in  $\RES(s)$.  
We prove that $\BinPHP^m_n$  (Theorem \ref{thm:phpubtreelike}) can be refuted in size $2^{O(n)}$  in treelike $\RES(1)$.  Our upper bound is contrasting with the unary case of the Pigeonhole  Principle, $\pPHP^m_n$, which instead requires  treelike $\RES(1)$ refutations of size $2^{\Omega(n \log n)}$, as proved in \cite{BeyersdorffGL10, DantchevR01}. %\TODO{Can we say more on optimality of our lower bounds for in $\RES(s)$ ?}

 As for the $k$-Clique principle,  also for the Pigeonhole Principle, we can prove that short $\RES(\log n)$ refutations for $\BinPHP^m_n$ can be efficiently obtained from short $\RES(1)$ of $\pPHP^m_n$ (Lemma \ref{sec:wphp}). Hence another observation raising from our lower bound  
 concerns the result of Buss and Pitassi in \cite{BussP97}, who proved  a quasipolynomial upper bounds (in the number of variables of $\pPHP^m_n$) for the size of refuting $\pPHP^m_n$ when $m \geq 2^{\sqrt{n\log n}}$. Indeed, they give the subexponential-in-$n$ upper bound of $2^{O(\sqrt{n\log n})}$. Hence no exponential-in-$n$ lower bound is possible in Resolution when $m \geq 2^{\sqrt{n\log n}}$.   Since we prove that  $\BinPHP^m_n$ requires $2^{n^{1-\delta}}$ size in 
 $\RES(s)$ for any $m>n$, then Lemma \ref{sec:wphp} is indicating that  Buss and Pitass's  result in \cite{BussP97} is essentially tight and cannot be 
 be proved for the binary version of the Pigeonhole principle.

\subsubsection{Contrasting  unary and binary principles}
\label{subsec:theory}
To work with a more general theory in which to contrast the complexity of refuting the binary and unary versions of combinatorial principles, following Riis \cite{SorenGap} we consider principles which are expressible as first order formulas with no finite model  in $\Pi_2$-form, i.e.  as $\forall \vec x \exists \vec w \varphi(\vec x,\vec w)$ where $\varphi(\vec x,\vec y)$ is a formula built on a family of relations $\vec R$. 
For example  the Ordering Principle, which states that a finite partial order has a maximal element is one of such principle. 
Its negation can be expressed in $\Pi_2$-form as:  
\[\forall x,y,z \exists w \ \neg R(x,x) \wedge (R(x,y) \wedge R(y,z) \rightarrow R(x,z)) \wedge R(x,w).\]
This can be translated into a unsatisfiable CNF $\OP$ using a {\em unary encoding} of the witness, as shown below.
In Definition \ref{def:binC} we explain how to generate a binary encoding $\BC$ from any  combinatorial principle $\CC$ expressible as  
a first order formulas in  $\Pi_2$-form with no finite models and  whose unary encoding we denote by $\UC$. 
For example  $\bOP$ would be the conjunction of the clauses below.
\begin{small}
\[
\begin{array}{lll}
	\quad \OP: \mbox{\em\underline{Unary encoding}} & \quad & \quad  \bOP: \mbox{\em \underline{Binary encoding}} \\
	\begin{array}{ll}
		\overline v_{x,x} & x \in [n] \\
		\overline v_{x,y} \vee \overline v_{y,z} \vee v_{x,z} & x,y,z \in [n] \\
		\bigvee_{i\in [n]} v_{x,i} & x \in [n]
	\end{array} 
	&
	\quad \quad \quad \quad \quad \quad\quad 
	&
	\begin{array}{ll}
		\overline \nu_{x,x} & x \in [n] \\
		\overline \nu_{x,y} \vee \overline \nu_{y,z} \vee \nu_{x,z} & x,y,z \in [n] \\
		\bigvee_{i\in [\log n]} \omega^{1-a_i}_{x,i} \vee \nu_{x,a} & x,a \in [n] \\
		\mbox{$a_1\ldots a_{\log n}$  binary representation of $a$} \\
		 \omega^{a_j}_{x,j}=\left\{
   		\begin{array}{ll}
            		\omega_{x,j} &  a_j=1 \\
            		\overline \omega_{x,j} & a_j=0
    		\end{array} \right. 
	\end{array} 
\end{array}
\]
\end{small}

As a second example we consider  the Pigeonhole Principle which  states  that a total mapping from $[m]$ to 
$[n]$ has necessarily a collision when $m$ and $n$ are integers with  $m>n$. Following Riis \cite{SorenGap} the negation of its  relational form can be expressed as a $\Pi_2$-formula as 
$$\forall x,y,z \exists w \ \neg R(x,0) \wedge (R(x,z) \wedge R(y,z) \rightarrow x=y) \wedge R(x,w)$$ 
and its usual unary and  binary propositional  encoding are:

\begin{small}
\[
\begin{array}{lll}
	 \pPHP: \mbox{\em\underline{Unary encoding}} & \quad &\quad  \BinPHP: \mbox{\em \underline{Binary encoding}} \\
	\begin{array}{ll}
		 \bigvee_{j=1}^{n}v_{i,j}  & i \in [m] \\
		  \overline v_{i,j} \vee \overline v_{i',j} & i,\not = i'\in [m], j\in [n]
	\end{array}
	&
\quad \quad \quad \quad \quad \quad\quad
	&
	\begin{array}{ll}
		 \bigvee_{j=1}^{\log n}\overline \omega_{i,j}  \vee  \bigvee_{j=1}^{\log n}\overline  \omega_{i',j}& i \not = i' \in [m] \\
	\end{array}
	
\end{array}
\]
\end{small}

Notice that in the case of Pigeonhole Principle, the existential witness $w$ to the type \emph{pigeon} is of the distinct type \emph{hole}. Furthermore, pigeons only appear on the left-hand side of atoms $R(x,z)$ and holes only appear on the right-hand side. 
For the Ordering Principle instead, the transitivity axioms effectively enforce the type of $y$ appears on both the left- and right-hand side of atoms $R(x,z)$. This account for why, in the case of the Pigeonhole Principle, we did not need to introduce any new variables to give the binary encoding, yet for the Ordering Principle a new variable $w$ appears. In Section \ref{sec:gentheo} we show that binary encodings are most interesting to study for $\Pi_2$ combinatorial principles \emph{all of whose witnesses are of a different type from the variables they are witnesses for.} %\red{\bf [Barny, I think you should write something formal  on this in the last Section and say more precisely what you mean ! I do not understand this]} 

In Section \ref{sec:gentheo} we  observe that  Lemma \ref{lem:reslog} and \ref{lem:php-un-bin} work also for the general case of $\UC$ and $\BC$ (Lemma \ref{lem:resloggen}). We also prove in Lemma \ref{def:binC} that  the usual binary encoding  $\BinPHP$ of the $\pPHP$ (\cite{DBLP:journals/siamcomp/FilmusLNRT15,DBLP:journals/jacm/BonacinaG15}) is provably equivalent in Resolution  to  the version of the binary version Pigeonhole principle defined from our translation to binary of Definition \ref{def:binC}.  We finally propose a framework to compare  lower bounds for the $\BC$ in $\RES(s)$ with lower bounds for $\UC$ in $\RES(1)$.

\subsubsection{Total comparisons and Linear Ordering principles}
 $\LOP$ formulae took on a certain importance in Resolution. In their more general form they were used in \cite{DBLP:journals/cc/BonetG01} to  prove the optimality of the 
size-width tradeoffs for Resolution (see \cite{Ben-sasson99shortproofs}). More importantly for this work, a modification of the $\LOP$ formulas ($\RLOP$) were used in \cite{AJPU} to exhibit a family of formulas exponentially separating proof size in read-once Resolution from Resolution. 

%\red{\bf To compare our lower bounds  for (binary) principles} in $\RES(s)$ to lower bounds for  (unary)  principles in Resolution,  
We  study under what conditions the complexity of proofs in Resolution will not increase significantly (by more than a polynomial factor) when shifting from the unary encoding to the binary encoding.  In Lemma \ref{lem:totoord} we prove that this is true for the negation of principles expressible as first order formula in $\Pi_2$-form involving {\em total variable comparisons}. Hence in particular (see Corollary \ref{cor:ale}) the binary version of the Linear Ordering  principle $\bLOP$ and its modification  $\bRLOP$ which separates read-once Resolution from Resolution (see \cite{AJPU}) are  polynomially provable in Resolution. 
It is worthy to notice that   $\bRLOP$ is polynomially provable in $\RES(\frac{1}{2}\log \log n)$, where we prove a lower bound for  $\BinC(G)$.
 %allows us to compare, in a way we make precise in Section \ref{sec:gentheo}, the lower bound for $\UnC$ in read-once Resolution of \cite{DBLP:conf/stoc/AtseriasBRLNR18} with our our lower bound for $\BinC$ in $\RES(\frac{1}{2}\log \log n)$.  

Finally, we also prove that the binary encoding of the general Linear Ordering $\bOP$ principle, where antisymmetry -- which entails total comparisons -- is not encoded, is also polynomially provable in Resolution. Ordering 
Principles are typically used to provide hierarchy separations (see for instance \cite{DBLP:journals/siamcomp/SegerlindBI04,DantchevR03}) inside Resolution-based proof systems.  Hence, loosely speaking, they mark the maximal border of what is still provable efficiently in a given proof system%\red{\bf [I do not like too much what I wrote.  Can we say this better ?]}
.
For this reason the upper bounds explained in this subsection for the binary version of the ordering principles should be interpreted
broadly speaking,   as saying that shifting to the binary encodings is not destroying the hardness of a unary principle when working in Resolution and hence binary encodings of combinatorial principles are still meaningful benchmarks to prove lower bounds for. 

%\TODO{Add A subsection explaing how to precisely we compare the lower bounds }
\subsubsection{Binary encodings of principles versus their Unary functional encodings}

The \emph{unary functional} encoding of a combinatorial principle replaces the big disjunctive clauses of the form $v_{i,1} \vee \ldots \vee v_{i,n}$, with $v_{i,1} + \ldots + v_{i,n} = 1$, where addition is made on the natural numbers. This is equivalent to augmenting the axioms $\neg v_{i,j} \vee \neg v_{i,k}$, for $j \neq k \in [n]$. One might argue that the unary functional encoding is the true unary analog to the binary encoding, since the binary encoding naturally enforces that there is a single witness alone. It is likely that the non-functional formulation was preferred for its simplicity (similarly as the Pigeonhole Principle is often given in its non-functional formulation).

In Subsection \ref{subsec:func}, we prove that the Resolution refutation size increases by only a quadratic factor when moving from the binary encoding to the unary functional encoding. This is interesting because the same does not happen for treelike Resolution, where the unary encoding has complexity $2^{\Theta(n \log n)}$  \cite{BeyersdorffGL10, DantchevR01}, while, as we prove in Subsection \ref{subsec:treephp} (Theorem \ref{thm:phpubtreelike}), the unary (functional) encoding is $2^{\Theta(n)}$. The unary encoding complexity is noted in \cite{DantchevR03} and remains true for the unary functional encoding with the same lower-bound proof. The binary encoding complexity is addressed directly in this paper.  

\subsection{Techniques and Organization}
The method of random restrictions in Proof Complexity is often employed  to prove size lower bounds. Loosely speaking the method works as follows:  we consider formulae having a given specific combinatorial property $P$; after hitting, with a suitable random partial assignment, on an allegedly short  proof of the formula we are refuting, we are left to prove that with high probability a formula with property $P$ is killed away from the proof. The growth rate as the probability approaches to 1 together with  a counting argument  using averaging (as the union bound), implies  a lower bound on the number of formulae with property P in the proof.  Lower bounds in $\RES(s)$ using random restrictions were known only for $s=2$ (see \cite{DBLP:journals/iandc/AtseriasBE02}).
Using a weak form of the Switching Lemma, lower bounds for $\RES(s)$ were obtained in \cite{DBLP:journals/siamcomp/SegerlindBI04,DBLP:journals/cc/Alekhnovich11}. From the latter paper we use the notion of
{\em covering number} of a $k$-DNF $F$, i.e. the minimal size of a set  of variables to hit all the $k$-terms in $F$.    
In this work we merge the covering number with the random restriction method together with an inductive argument on the $s$, to get size lower bounds in $\RES(s)$ specifically for binary encoding of combinatorial principles. 

After a section with the preliminaries, the paper is divided into four sections: one with the lower bound for the $k$-Clique Principle, one containing all the results  for the (weak) Pigeohole principle, one for the contrasting the proof complexity between unary and binary principles containing all the results 
about the various Ordering Principles, and finally the last section containing  a general  approach to unary vs binary  encodings for principle expressible as  
a  $\Pi_2$ formulae. % and the comparison of our lower bound for $\BinC$  with that for $\UnC$ of \cite{DBLP:conf/stoc/AtseriasBRLNR18}.

\section{Preliminaries}
\label{sec:preliminaries}

%\subsection*{Resolution and Res($s$).}

We denote by $\top$ and $\bot$ the Boolean values {}``true'' and
{}``false'', respectively. A \emph{literal} is either a propositional
variable or a negated variable. We will denote literals by small
letters, usually $l$'s. An $s$\emph{-conjunction} ($s$-\emph 
{disjunction})
is a conjunction (disjunction) of at most $k$ literals. A {\em clause} with $s$ literals is a $s$-disjunction. The width
$w(C)$ of a clause $C$ is the number of literals in $C$. A \emph{term}
($s$\emph{-term}) is either a conjunction ($s$-conjunction) or a
constant, $\top$ or $\bot$. A $s$-\emph{DNF} or $s$\emph{-clause}
($s$\emph{-CNF}) is a disjunction (conjunction) of an unbounded number
of $s$-conjunctions ($s$-disjunctions). We will use calligraphic
capital letters to denote $s$-CNFs or $s$-DNFs, usually ${\cal C}$s for CNFs, ${\cal D}$s for DNFs and ${\cal F}$s for both.
% Sometimes,
%when clear from the context, we will say {}``clause'' instead of
%{}``$k$-clause'', even though, formally speaking, a clause is a
%$1$-clause.

We can now describe the propositional refutation system $\RES\left(s\right)$ (\cite{krabook95}). It is used \emph{to
refute} (\mbox{i.e.} to prove inconsistency) of a given set of $s$-clauses
by deriving the empty clause from the initial clauses. There are four
derivation rules:

\begin{enumerate}
\item The $\wedge$-\emph{introduction rule} is\[
\frac{\mathcal{D}_{1}\vee\bigwedge_{j\in J_{1}}l_{j}\quad\mathcal{D}_ 
{2}\vee\bigwedge_{j\in J_{2}}l_{j}}{\mathcal{D}_{1}\vee\mathcal{D}_{2} 
\vee\bigwedge_{j\in J_{1}\cup J_{2}}l_{j}},\]
provided that $\left|J_{1}\cup J_{2}\right|\leq s$.
\item The \emph{cut (or resolution) rule} is\[
\frac{\mathcal{D}_{1}\vee\bigvee_{j\in J}l_{j}\quad\mathcal{D}_{2}\vee 
\bigwedge_{j\in J}\neg l_{j}}{\mathcal{D}_{1}\vee\mathcal{D}_{2}},\]

\item The two \emph{weakening rules} are\[
\frac{\mathcal{D}}{\mathcal{D}\vee\bigwedge_{j\in J}l_{j}}\quad\textrm 
{and}\quad\frac{\mathcal{D}\vee\bigwedge_{j\in J_{1}\cup J_{2}}l_{j}} 
{\mathcal{D}\vee\bigwedge_{j\in J_{1}}l_{j}},\]
provided that $\left|J\right|\leq s$.
\end{enumerate}
A $\RES(s)$ refutation can be considered as a directed acyclic
graph (DAG), whose sources are the initial clauses, called also axioms,
and whose only sink is the empty clause. We shall define \emph{the
size of a proof} to be the number of the internal nodes of the graph,
i.e. the number of applications of a derivation rule, thus ignoring
the size of the individual $s$-clauses in the refutation.

In principle the $s$ from {}``$\RES(s)$'' could depend
on $n$ --- an important special case is $\RES(\log n)$
%$\COM{\TODO{We should be more precise here on the role and dependencies of $s$ from  $n$. Neil had somethings to say !}}.

Clearly, $\RES(1)$ is \emph{(ordinary) Resolution}, working
on clauses, and using only the cut rule, which becomes the usual
resolution rule, and the first weakening rule. Given an unsatisfiable CNF ${\cal C}$, 
and a $\RES(1)$ refutation $\pi$ of ${\cal C}$ the width of $\pi$, $w(\pi)$ is the  maximal width of a  clause in $\pi$.
The width  refuting ${\cal C}$ in Res(1), $w(\vdash {\cal C})$, is the minimal width over all $\RES(1)$ refutations of ${\cal C}$.

A {\em covering set} for a $s$-DNF ${\cal D}$  is a set of literals $L$ such that each term of ${\cal D}$ has for at least a literal in $L$.
The {\em covering number} $c({\cal D})$ of a $s$-DNF ${\cal D}$ is the minimal size of a covering set for ${\cal D}$.  

Let ${\cal F}(x_1\ldots,x_n)$ be a boolean $s$-DNF (resp. $s$-CNF) defined over 
variables $X=\{x_1,\ldots,x_n\}$. A  {\em partial assignment} $\rho$ to ${\cal F}$  is a truth-value assignment to some 
of the variables of ${\cal F}$: $dom(\rho) \subseteq X$.  By ${\cal F}\!\!\!\upharpoonright_\rho$ we denote the formula ${\cal F}'$ over variables in $X\setminus dom(\rho)$ 
obtained from ${\cal F}$ after simplifying in it the variables in $\mathit{dom}(\rho)$ according to the usual boolean simplification rules of clauses and terms.

\subsection{$\RES(s)$ vs  Resolution}

Similarly to what was done for treelike $\RES(s)$ refutations in \cite{EGM}, if we turn a $\mbox{Res}\left(s\right)$ refutation of a given set
of $s$-clauses $\Sigma$ upside-down, \mbox{i.e.} reverse the edges of the
underlying graph and negate the $s$-clauses on the vertices, we get
a special kind of restricted branching $s$-program. The restrictions are
as follows.

Each vertex is labelled by a $s$-CNF which partially represents the  
information
that can be obtained along any path from the source to the vertex (this is a \emph{record} in the parlance of \cite{proofs_as_games}).
Obviously, the (only) source is labelled with the constant $\top$.
There are two kinds of queries, which can be made by a vertex:

\begin{enumerate}
\item Querying a new $s$-disjunction, and branching on the answer, which
can be depicted as follows.\begin{equation}
\begin{array}{ccccc}
  &  & \mathcal{C}\\
  &  & ?\bigvee_{j\in J}l_{j}\\
  & \top\swarrow &  & \searrow\bot\\
\mathcal{C}\wedge\bigvee_{j\in J}l_{j} &  &  &  & \mathcal{C}\wedge 
\bigwedge_{j\in J}\neg l_{j}\end{array}\label{eq:new-query}\end 
{equation}

\item Querying a known $s$-disjunction, and splitting it according to  
the answer:
\begin{equation}
\begin{array}{ccccc}
  &  & \mathcal{C}{\wedge\bigvee}_{j\in J_{1}\cup J_{2}}l_{j}\\
  &  & ?\bigvee_{j\in J_{1}}l_{j}\\
  & \top\swarrow &  & \searrow\bot\\
\mathcal{C}\wedge\bigvee_{j\in J_{1}}l_{j} &  &  &  & \mathcal{C} 
\wedge\bigvee_{j\in J_{2}}l_{j}\end{array}
\label{eq:split-query}
\end{equation}

\end{enumerate}

\noindent There are two ways of forgetting information,
{\begin{equation}
\begin{array}{c}
{\cal C}_{1}\wedge{\cal C}_{2}\\
\downarrow\\
{\cal C}_{1}\end{array}\qquad\textrm{and}\qquad\begin{array}{c}
\mathcal{C}\wedge\bigvee_{j\in J_{1}}l_{j}\\
\downarrow\\
\mathcal{C}\wedge\bigvee_{j\in J_{1}\cup J_{2}}l_{j}\end{array},\label 
{eq:forget}\end{equation}
the point being that forgetting allows us to equate the information
obtained along two different branches and thus to merge them into
a single new vertex.} A sink of the branching $s$-program must be labelled
with the negation of a $s$-clause from $\Sigma$. Thus the branching $s$-program
is supposed by default to solve the \emph{Search problem for $\Sigma$}:
given an assignment of the variables, find a clause which is falsified
under this assignment.

The equivalence between a $\mbox{Res}\left(s
\right)$
refutation of $\Sigma$ and a branching $s$-program of the kind above is
obvious. Naturally, if we allow querying single variables only, we
get branching $1$-programs -- decision DAGs -- that correspond to Resolution. If we do not
allow the forgetting of information, we will not be able to merge distinct
branches, so what we get is a class of decision trees that correspond
precisely to the treelike version of these refutation systems.

Finally, we mention that the queries of the form (\ref{eq:new-query})
and (\ref{eq:split-query}) as well as forget-rules of the form (\ref 
{eq:forget})
give rise to a Prover-Adversary game (see \cite{proofs_as_games}
where this game was introduced for Resolution). In short, Adversary
claims that $\Sigma$ is satisfiable, and Prover tries to expose him.
Prover always wins if her strategy is kept as a branching program
of the form we have just explained, whilst a good (randomised)  
Adversary's
strategy would show a lower bound on the branching program, and thus
on any $\mbox{Res}\left(k\right)$ refutation of $\Sigma$.

\begin{lemma}
If a CNF $\phi$ has a refutation in $\RES(k+1)$ of size $N$, whose corresponding branching $(k+1)$-program has no records of covering number $\geq d$, then $\phi$ has a $\RES(k)$ refutation of size $2^d\cdot N$.
\label{lem:covering-number}
\end{lemma}

\begin{proof}
In the branching program, consider a $(k+1)$-CNF record $\phi$ whose covering number $\leq d$ is witnessed by variable set $V':=\{v_1,\ldots,v_d\}$. Now in place of the record $\phi$ we expand a tree of size $2^d$ questioning all the variables of $V'$. Each evaluation of these reduces $\phi$ to a $k$-CNF that logically implies $\phi$. 
\end{proof}

\section{The binary encoding of $k$-Clique}
\label{sec:k-clique}

Consider a graph $G$ such that $G$ is formed from $k$ blocks of $n$ nodes each: $G =(\bigcup_{b\in[k]} V_b, E)$, where edges may only appear between distinct blocks. Thus, $G$ is a $k$-partite graph. Let the 
edges in $E$ be denoted as pairs of the form $E((i,a),(j,b))$, where $i\neq j \in[k]$ and $a,b\in [n]$.

The (unary) $k$-Clique CNF formulas $\UnC(G)$ for $G$, has variables $v_{i,q}$ with $i \in [k], a \in [n]$, with clauses $\neg v_{i,a} \vee \neg v_{j,b}$ whenever $\neg E((i,a),(j,b))$ (\mbox{i.e.} there is no edge between node $a$ in block $i$ and node $b$ in block $j$), and clauses $\bigvee_{a \in [n]} v_{i,a}$, for each block $i$. This expresses that $\mathcal{G}^n_k$ has a $k$-clique, which we take to be a contradiction, since we will arrange for $G$ not to have a $k$-clique.

\medskip
 $\BinC(G)$  variables $\omega_{i,j}$ range over $i \in [k], j \in [\log n]$.
Let $a \in [n]$ and let $a_1 \ldots  a_{\log n}$ be its binary representation.  
Each (unary) variable $v_{i,j}$ semantically corresponds to the  conjunction $(\omega^{a_1}_{i,1}\wedge \ldots \wedge \omega^{a_{\log n}}_{i,\log n})$,
where
$$
\omega^{a_j}_{i,j}=
\left\{
\begin{array}{ll}
\omega_{i,j} & \mbox{ if $a_j=1$} \\
\overline \omega_{i,j} & \mbox{ if $a_j=0$}
\end{array}
\right.
$$
Hence in $\BinC(G)$ we encode the unary clauses $\neg v_{i,a} \vee \neg v_{j,b}$,
by the clauses
$$
(\omega^{1-a_1}_{i,1}\vee \ldots \vee \omega^{1-a_{\log n}}_{i,\log n}) \vee (\omega^{1-b_1}_{j,1}\vee \ldots \vee \omega^{1-b_{\log n}}_{j,{\log n}})
$$

%A disjunct of literals $(-1)^{a_1}v_{i,1} \vee \ldots \vee (-1)^{a_{\log n}}v_{i,\log n}$ encodes a binary number of length $\log n$ through its negated form $(-1)^{a_1+1}\ldots (-1)^{a_{\log n}+1}$, where $(-1)^{a_j+1}$ is mapped to $1$, if it is $1$, and to $0$ otherwise. Let $(-1)^{a_1}v_{i,1} \vee \ldots \vee (-1)^{a_{\log n}}v_{i,\log n}$ give associated number $a$, and $(-1)^{b_1}v_{\ell,1} \vee \ldots \vee (-1)^{b_{\log n}}v_{\ell,\log n}$ give associated number $b$. 

%The binary $k$-Clique formula, $\BinC(G)$, has variables $v_{i,j}$ with $i \in [k], j \in [\log n]$, with clauses \[ (-1)^{a_1}v_{i,1} \vee \ldots \vee (-1)^{a_{\log n}}v_{i,\log n} \vee (-1)^{b_1}v_{\ell,1} \vee \ldots \vee (-1)^{b_{\log n}}v_{\ell,\log n}, \] whenever  $\neg E((i,a),(\ell,b))$.

By the next Lemma short Resolution refutations for  $\UnC(G)$ can be translated into short $\RES(\log n)$ refutations of $\BinC(G)$. 
hence to obtain lower bounds for $\UnC(G)$ in Resolution, it suffices to obtain lower bounds for 
 $\BinC(G)$ in $\RES(\log n)$. 

\begin{lemma}
\label{lem:reslog}
Suppose there are Resolution refutations of $\UnC(G)$ of  size $S$. 
Then there are $\RES(\log n)$ refutations of $\BinC(G)$  of size $S$.
\end{lemma}
\begin{proof}
Where the decision DAG for $\UnC(G)$ questions some variable $v_{i,a}$, the decision branching $\log n$-program questions instead 
$(\omega^{1-a_1}_{1,1}\vee \ldots \vee \omega^{1-a_{\log n}}_{1,{\log n}})$ where the out-edge marked true in the former becomes false in the latter, and vice versa. What results is indeed a decision branching $\log n$-program for $\BinC(G)$, and the result follows.
\end{proof}

Following  \cite{Beyersdorff:2013:TOCL,DBLP:conf/stoc/AtseriasBRLNR18,DBLP:journals/combinatorica/LauriaPRT17} we consider $\BinC(G)$ formulas  where $G$ is   a random graph distributed according to a variation of the Erd\"os-R\'enyi as defined in  \cite{Beyersdorff:2013:TOCL}. In the standard  model, random graphs on $n$ vertices are constructed by including every edge independently with probability $p$. It is known that $k$-cliques appear
at the threshold probability $p^*=n^{- \frac{2}{k-1}}$. If $p < p^*$, then with high probability there is
no $k$-clique.
By  $\mathcal{G}^n_{k,\epsilon}(p)$ we denote the  distribution on random multipartite Erd\H{o}s-Renyi graph with $k$ blocks $V_i$ of $n$ vertices each, where each edge is present with probability $p$ depending on $\epsilon$. 
For $p=n^{- (1+\epsilon)\frac{2}{k-1}}$ we just write $\mathcal{G}^n_{k,\epsilon}$.  

We use the notation
$G =(\bigcup_{b\in[k]} V_b, E) \sim \mathcal{G}^n_k(p)$ to say that $G$ is a graph drawn at random from the distribution $\mathcal{G}^n_k(p)$.

In the next section we explore lower bounds for  $\BinC(G)$ in Res($s$) for $s \geq 1$, when $G\sim \mathcal{G}^n_k(p)$.
%For $n=2^k$, $\mathcal{G}^n_k$ has with high probability no $k$-clique but many $\frac{k}{3}$-cliques (\red{Cit.}).

\medskip
\subsection{Res(s) lower bounds for $\BinC$}
 Let $\alpha$ be a constant such that $0<\alpha<1$. Define a set of vertices $U$ in $G$, $U \subseteq V$ to be an {\em $\alpha$-transversal} if: (1) $|U| \leq \alpha k$, and (2) for all $b\in [k]$,  $|V_b\cap U|\leq 1$.
 Let $B(U)\subseteq [k]$ be the set of blocks mentioned in $U$, and let $\overline{B(U)} = [k]\setminus B(U)$. 
 We say that $U$ is {\em extendible}  in a block $b \in \overline{B(U)}$ if there exists a vertex $a \in V_b$ which is a common
  neighbour of all nodes in $U$, i.e. $a \in N_c(U)$ 
 where $N_c(U)$ is the set of {\em common neighbours} of vertices in $U$ i.e. $N_c(U)=\{v \in V \;|\; v \in \bigcap_{u\in U}N(u)\}$.  
\medskip

Let $\sigma$ be a partial assignment (a restriction) to the variables of $\BinC(G)$ and $\beta$ a constant such that  $0<\beta<1$. We call $\sigma$,  {\em $\beta$-total} if $\sigma$ assigns $\lfloor \beta \log n\rfloor $ bits in each block $b\in [k]$, \mbox{i.e.} $\lfloor \beta \log n\rfloor $ variables $\omega_{b,i}$  in each block $b$. 
Let $v=(i,a)$  be the $a$-th node in the $i$-the block in  $G$. 
 We say that a restriction $\sigma$ is {\em consistent} with $v$  if for all $j\in [\log n]$, $\sigma(\omega_{i,j})$ is either $a_j$ 
 %\red{[Should this instead be $1 -a_j$??? ]}
or not assigned.

\begin{definition} Let $0<\alpha,\beta <1$.
A $\alpha$-transversal set of vertices $U$ is $\beta$-extendible, if for 
all $\beta$-total restriction $\sigma$, there is a node $v^b$ in each block $b \in \overline{B(U)}$, such that 
$\sigma$ is consistent with $v^b$.
\end{definition}

\begin{lemma} (Extension Lemma)
\label{lem:extension}
Let $0<\epsilon <1$, let $k\leq \log n$.  Let $1>\alpha >0$ and $1>\beta>0$ such that $1-\beta >\alpha(2+\epsilon)$.
Let $G \sim \mathcal{G}^n_{k,\epsilon}$. With high probability both the following properties hold:

\begin{enumerate}
\item all $\alpha$-transversal sets $U$ are $\beta$-extendible; 
\item $\mathcal{G}$ does not have a $k$-clique.
 \end{enumerate}
\label{lem:stefan}
\end{lemma}
\begin{proof}
Let $U$ be an $\alpha$-transversal set and $\sigma$ be a $\beta$-total restriction. The probability that a vertex $w$ is  in $N_c(U)$ is $p^{\alpha k}$. Hence $w 	\not \in N_c(U)$  with probability  
$(1-p^{\alpha k})$. After $\sigma$ is applied, in each block $b \in \overline{B(U)}$ remain $2^{\log n - \beta \log n}=
n^{1-\beta}$ available vertices. Hence the probability that we cannot extend $U$ in each block of  $\overline{B(U)}$ after $\sigma$ is applied is $(1-p^{\alpha k})^{n^{1-\beta}}$.  Fix $c=2+\epsilon$ and $\delta=1-\beta- \alpha c$. Notice that $\delta>0$ by our choice of $\alpha$ and $\beta$.  Since $p=\frac{1}{n^{\frac{c}{k}}}$, previous probability is $(1-1/n^{\alpha c})^{n^{1-\beta}}$, which is asymptotically $e^{-\frac{n^{1-\beta}}{n^{\alpha c}}}= e^{-n^{\delta}}$ .

There are ${k \choose \alpha k}$ possible $\alpha$-transversal sets $U$ and  ${\log n \choose \beta \log n} \cdot k$ possible 
$\beta$-total restrictions $\sigma$. 

$$
\begin{array}{lll}
{k \choose \alpha k} \cdot {\log n \choose \beta \log n} \cdot k & \leq  k^{\alpha k} \cdot (\log n)^{\beta \log n} \cdot k\\ 
& = 2^{\alpha k \log k + \beta \log n \log \log n + \log k} \\
& \leq 2^{\log^2 n}
\end{array}
$$
Notice that the last inequality holds since $k\leq \log n$. 
Hence the probability  that there is in $G$ no $\alpha$-transeversal set $U$ which is $\beta$-extendible is going to $0$ as 
$n$ grows.

\medskip

To bound the probability that $\mathcal{G}$ contains a $k$-clique, notice that the expected number of  $k$ cliques is  ${n \choose k} \cdot p^{{k \choose 2}} \leq n^k \cdot p^{(k(k-1)/2)}$. Recalling $p=1/n^{c/k}$, we get 
that the probability that $G$ does not have a $k$-clique is $n^k \cdot n^{-c(k-1)/2} = n^{k-c(k-1)/2}$. Since $c=2+\epsilon$,  
$k-c(k-1)/2= 1-\frac{\epsilon}{2}(k-1)$.  Hence $n^k \cdot n^{-c(k-1)/2} \leq 2^{-\log n}$ for sufficiently large $n$ and since $k\leq \log n$.

So the probability that either property (1) or (2) does not hold is bounded above by  $2^{\log^{2}n}\cdot e^{-n^{\delta}} + 2^{-\log^{2}n}$ which is below $1$ for sufficiently large $n$. 
\end{proof}
%We refer to the first  item of Lemma~\ref{lem:stefan} as an \emph{extension property}, which we can use to argue width lower bounds for Resolution for $\BinC(G)$.

\medskip

Let $s\geq 1$ be an integer. Call a $\frac{1}{2^{s+1}}$-total assignment  to the variables of $\BinC(G)$ an {\em $s$-restriction}.   
 A  \emph{random $s$-restriction} for $\BinC(G)$ is an $s$-restriction obtained by choosing independently in each block $i$, $ \lfloor \frac{1} {2^{s+1}} \log n\rfloor $ variables among $\omega_{i,1},\ldots,\omega_{i,\log n}$, and seting these uniformly at random to $0$ or $1$.  

%In particular  a $\gamma$-random restriction is a $\gamma$-restriction.   Let $s$ be an integer, $s\geq1$. %We are interested in  applying  $\gamma$-total restrictions to  $\BinC(G)$, where $\gamma= 1/2^{s+1}$.  

\medskip
Let $s,k\in \mathbb N $, $s,k\geq1$ and let $G$ be graph over $nk$ nodes and $k$ blocks which does not contain a $k$-clique.
Consider the following property.
\begin{definition} (Property $\Cl(G, s,k)$).
For any $s$-restriction $\rho$, there are no Res($s$) refutations of 
$\BinC(G)\!\!\!\upharpoonright_\rho$ %\red{[DEFINE NOTATION OF RESTRICTED FORMULA]}
of size less $n^{\frac{k-1}{24^2 s}}$.
\end{definition}

If property $\Cl(G,s,k)$ holds, we immediately  have $n^{\Omega(k)}$ \sloppy size lower bounds for refuting $\BinC(G)$ in $Res(s)$.
%In the case of random $G\sim \mathcal{G}^n_k(p)$, we have:
\begin{corollary}
\label{cor:BinC}
Let $s,k$ be integers, $s\geq1,k>1$. Let $G$ be a graph and assume that  $\Cl(G,s,k)$ holds. %with high probability. 
Then there are no Res($s$) refutations of   $\BinC(G)$ of size smaller that $n^{\frac{k-1}{24^2 s}}$. 
\end{corollary}
\begin{proof}
 For $\rho$ the empty assignment there are no Res($s$) refutations of 
$\BinC(G)$ of size smaller than   $n^{\frac{k-1}{24^2 s}}$.
\end{proof}

We use the previous corollary to prove lower bounds for $\BinC(G)$ in $\RES(s)$ as long as $s \leq \frac{1}{2} \log \log n$.
\begin{theorem}
\label{thm:main-BinC}
Let $0<\epsilon<1$ be given.
Let $k$ be an integer with $k> 1$.
Let $s$ be an integer with  $1< s \leq \frac{1}{2} \log \log n$.
Then there exists a graph $G$ such that Res($s$) refutations of $\BinC(G)$ have size $n^{\Omega(k)}$. 
\end{theorem}

\begin{proof}
By Lemma~\ref{lem:extension}, we can fix $G\sim \mathcal{G}^n_{k,\epsilon}$ such that: 
\begin{enumerate}
\item all $\alpha$-transversal sets $U$ are $\beta$-extendible; 
\item $\mathcal{G}$ does not have a $k$-clique.
 \end{enumerate}
 We will prove, by induction on $s \leq \frac{1}{2}  \log \log n$, that property $\Cl(s,k,G)$ does hold. The result then follows by Corollary \ref{cor:BinC}. 
 Lemma \ref{lem:Res1-bin-k-clique} is the base case and Lemma \ref{lem:ind} the inductive case.
\end{proof}

%The formula resulting form $\BinC(G)$ after such a restriction is  applied it is called $\Xi_s$.
%We will prove, by induction on $s$,  that for $G \sim \mathcal{G}^n_k(p)$ and for $s \leq \frac{1}{2}  \log \log n$, $\Cl(s,k,G)$ does hold w.h.p. 
%Next  Theorem is the base  case and uses  Lemma \ref{lem:extension}.

\begin{lemma} (Base Case)
\label{lem:Res1-bin-k-clique}
$\Cl(1,k,G)$ does hold.
%Let $0<\alpha,\beta<1$ and $k=\log^d n$, for some $d$. For all sufficiently large $n$, there exists some graph $\mathcal{G}=\mathcal{G}^n_k(p)$ such that for any restriction $\rho$ of size $\beta \log n$ in all the partitions, the principle $\Xi_1(n)$ obtained by considering $k$-Clique over $\mathcal{G}$ under the restriction $\rho$ is a contradiction and all Resolution refutations of $\Xi_1(n)$ are of size $\geq n^{\alpha k}$. 
%Let $n=2^k$ and $\alpha:=\frac{1}{6},\beta:=\frac{1}{6},\gamma:=\frac{1}{6}, c=1$. Let $\Xi$ be a set of clauses formed from the binary $k$-clique over $\mathcal{G}^n_k$ after it has been hit with a random restriction on $\frac{1}{3}\log n$ of the bits of each of its partitions. Then any Resolution refutation of $\Xi$ must have width at least $n^{\alpha k}$
\label{cor:Res1-bin-k-clique-bis}
%\medskip \red{TO DO: This is the old Corollary 1.  Write the  proof of this and use precisely the extension property}
\end{lemma}
\begin{proof}
Fix $\beta=\frac{3}{4}$ and $\alpha=\frac{1}{4(2+\epsilon)}\geq \frac{1}{12}$. Let $\rho$ be a $1$-restriction, that is a $\frac{1}{4}$-total assignment. 
We claim that any Resolution refutation of $\BinC(G)\!\!\!\upharpoonright_\rho$ must have width at least $\frac{k\log n}{24}$. This is a consequence of the extension property which allows Adversary to play against Prover with the following strategy: for each block, while fewer than $\frac{\log n}{2}$ bits are known, Adversary offers Prover a free choice. Once $\frac{\log n}{2}$ bits are set then Adversary chooses an assignment for the remaining bits according to the extension property. Since $\frac{1}{4}+\frac{1}{2}=\frac{3}{4}$, this allows the game to continue until some record has width at least $\frac{\log n}{2} \cdot \frac{k}{12}=\frac{k\log n}{24}$. Size-width tradeoffs for Resolution  \cite{Ben-sasson99shortproofs} tells us that minimal size to refute any unsat CNF $F$ is lower bounded by  $2^{(\mathit{w(\vdash F)-w(F))}^2/ \mathit{V(F)}}$. In our case $w(F)= 2\log n $, hence 
the minimal size required is $\geq 2^{\frac{(\frac{k \log n}{24} - 2 \log n)^2}{k \log n}}= 
2^{\frac{\log n (\frac{k}{24}-2)^2}{k}}= n^{\frac{(\frac{k}{24} -2)^2}{k}}$.  It is not difficult to see that $\frac{(\frac{k}{24} -2)^2}{k}\geq \frac{(k-1)}{24^2}$, the result is proved.
%that any Resolution refutation of $\BinC(G)\!\!\!\upharpoonright_\rho$ must have size at least $n^{\frac{k}{4}}$.
\end{proof}

%Let $\zeta(s)=\frac{2^{s^2+3s}-1}{2^{s^2+3s}}=(1-\frac{1}{2^{s^2+3s}})$ and note that $\ln \zeta(s)=-2^{s^2+3s}$  for positive integers $s$. 

\begin{lemma}(Inductive Case)
\label{lem:ind}
%Let $k$ be an integer with $k>1$ and let $0<\epsilon<1$. Let $s$ be an integer with  $1< s \leq \frac{1}{2} \log \log n$. Then 
$$\Cl(s-1,k,G) \mbox{ implies } \Cl(s,k,G).$$

\end{lemma}

\begin{proof}
We prove the contrapositive. 
Fix $\delta=1/24^2$. Let $\zeta(s)=(1-\frac{1}{2^{s^2+3s}})$ and $r=\frac{\delta (k-1) \log n}{s}$. 
Assume there is some $s$-restriction $\rho$ such that there exists a Res($s$) refutation $\pi$ of $\BinC(G)\!\!\!\upharpoonright_\rho$ with size less than $n^r$. %{\frac{\delta (k-1)\log n}{s}}$, with $\delta=1/24^2$. 
Notice that $n^r \leq 2^{-\log(\zeta(s))r}$. % \frac{\delta (k-1) \log n}{s}}$.
Let us call a \emph{bottleneck}, a record ${\cal R}$ in $\pi$ whose covering number is $\geq \delta (k-1) \log n$. 
In such a record it is always possible to find $r= \frac{\delta (k-1) \log n}{s}$ $s$-tuples of literals $T_1=(\ell^1_1,\ldots,\ell^s_1),\ldots,T_r=(\ell^1_r,\ldots,\ell^s_r)$ so that these $s$-tuples are pairwise disjoint (when considered a sets of size $s$)  such that the $\bigwedge T_i$'s are the terms of the $s$-DNF forming the record. By our size assumptions on $\pi$, there are  $\leq n^r$ bottlenecks.
Let $\sigma$ be a \emph{$s$-random restriction} on the variables of $\BinC(G)\!\!\!\upharpoonright_\rho$. 
Let us say that   {\em $\sigma$ kills a tuple $T$} if it sets to $0$ all literals in $T$ (notice that a record is the negation of $s$-DNF)
and that {\em $T$ survives $\sigma$} otherwise. 
And that {\em $\sigma$ kills ${\cal R}$} if it kills at at least one of the  tuples in ${\cal R}$. 
Let $\Sigma_i$ be the event that $T_i$ survives $\sigma$ and $\Sigma_{\cal R}$ the event that $R$ survives $\sigma$. 
We want to  prove that with high probability $\sigma$ kills all bottlenecks from $\pi$. 
We then study upper bounds on  $\Pr[\Sigma_R]$.  Since $T_1,\ldots,T_r$ are tuples in ${\cal R}$, then
$\Pr[\Sigma_R] \leq  \Pr[\Sigma_1 \wedge \ldots \wedge \Sigma_r]$. Moreover 
$\Pr[\Sigma_1 \wedge \ldots \wedge \Sigma_r]=\prod_{i=1}^r\Pr[\Sigma_i | \Sigma_1 \wedge \ldots \wedge \Sigma_{i-1}]$.

\begin{claim} For all $i=1,\ldots, r$,
$\Pr[\Sigma_i | \Sigma_1 \wedge \ldots \wedge \Sigma_{i-1}]\leq \Pr[\Sigma_i]$.
\end{claim}
\begin{proof}
We will prove that
$\Pr[\Sigma_i | \neg \Sigma_1 \vee \ldots \vee \neg \Sigma_{i-1}] \geq Pr[\Sigma_i]$. This gives the claim
using Lemma \ref{lem:probcond} (i). We claim that  for $i \not = j \in [r]$: 
\begin{eqnarray}
\Pr[\Sigma_i | \neg \Sigma_j] \geq \Pr[\Sigma_i] \label{Eq}
\end{eqnarray}
Hence repeated applications of Lemma \ref{lem:probcond} (ii), prove that
$\Pr[\Sigma_i | \neg \Sigma_1 \vee \ldots \vee \neg \Sigma_{i-1}] \geq Pr[\Sigma_i].$

To prove Equation \ref{Eq}, let $B(T_i)$ be the set of blocks mentioned in $T_i$. 
If $B(T_i)$ and $B(T_j)$ are disjoint, then clearly $\Pr[\Sigma_i | \neg \Sigma_j] = \Pr[\Sigma_i]$. 
When $B(T_i)$ and $B(T_j)$ are not disjoint,  we reason as follows:
For each $\ell \in B(T_i)$, let $T_i^\ell$ be the set of variables in $T_i$ mentioning block $\ell$.
$T_i$ is hence partitioned into $\bigcup_{\ell\in B(T_i)}T^\ell_i$ and hence the event "$T_i$ surviving  $\sigma$", 
can be partitioned into the sum of  the events that $T^\ell_i$ survives to $\sigma$, for  $\ell \in B(T_i)$.
Denote by $\Sigma^\ell_i$  the event  "$T^\ell_i$ survives $\sigma$" and
let A=$B(T_i)\cap B(T_j)$ and $B=B(T_i)\setminus (B(T_i)\cap B(T_j))$.  The following inequalities holds:

\begin{eqnarray}
\Pr[\Sigma_i |\neg \Sigma_j] &=&\Pr[\exists \ell \in B(T_i): \Sigma_i^\ell | \neg \Sigma_j] \\
& =& \sum_{\ell \in B(T_i)} \Pr[\Sigma_i^\ell | \neg \Sigma_j]\\
&=&\sum_{\ell \in A} \Pr[\Sigma_i^\ell  | \neg \Sigma_j] +\sum_{\ell \in B} \Pr[\Sigma_i^\ell  | \neg \Sigma_j]\\
\end{eqnarray}

Since $B$ is disjoint from $B(T_j)$, as for the case above for each $\ell \in B$, $\Pr[\Sigma_i^\ell  | \neg \Sigma_j]= \Pr[\Sigma_i^\ell]$. Then:
\begin{eqnarray}
\sum_{\ell \in B} \Pr[\Sigma_i^\ell  | \neg \Sigma_j] = \sum_{\ell \in B} \Pr[\Sigma_i^\ell] \\
\end{eqnarray}

Notice that  $T_i$ and $T_j$ are disjoint, hence knowing that some indices in blocks $\ell \in A$ are already chosen to 
kill $T_j$, only increase the chances of $T_i$ 
to survive (since less positions are left in the blocks  $\ell \in A$ to potentially kill $T_i$). 

Hence:
\begin{eqnarray}
\sum_{\ell \in A} \Pr[\Sigma_i^\ell  | \neg \Sigma_j] \geq \sum_{\ell \in A} \Pr[\Sigma_i^\ell] \\
\end{eqnarray}

Which proves the claim since:
\begin{eqnarray}
\sum_{\ell \in A} \Pr[\Sigma_i^\ell ] + \sum_{\ell \in B} \Pr[\Sigma_i^\ell]= \Pr[\Sigma_i] 
\end{eqnarray}

\end{proof}

Let $\gamma=1/2^{s+1}$. Lemma \ref{lem:survival-maximised-k-clique} below shows that,  
$\Pr[\Sigma_i] \leq 1-\frac{\gamma^s}{2^{2s}} \leq  \zeta(s)$, for all $i=1,\ldots, r$. Then by the Claim, 
$$\Pr[\Sigma_R] \leq \zeta(s)^r=n^{-r}.$$ 

Consider now   the restriction $\tau= \rho \sigma$. This  is a $(s-1)$-restriction on the variables of $\BinC(G)$.
Since there  are fewer than $n^r$ bottlenecks and $\Pr[\Sigma_R] \leq n^{-r}$, then by the union bound $\tau$ is a $(s-1)$-restriction 
that kills all bottlenecks of $\pi$.  Then, by Lemma~\ref{lem:covering-number}, we can morph $\pi$ through the restriction $\tau$ to a $\RES(s-1)$ refutation of  $\BinC(G)\!\!\!\upharpoonright_\tau$ of  size $2^{\frac{\delta (k-1) \log n}{s}} \cdot 2^{-\log(\zeta(s))\frac{\delta k \log n}{s}}=n^{\frac{\delta (k-1)}{s} (1-\log(\zeta(s)))}$. 
But this is smaller than $n^{\frac{\delta (k-1)}{s-1}}$ and this is contradicting $\Cl(s-1,k,G)$.

Notice that the previous argument can  be applied while $s<\frac{\gamma}{2}\cdot \log n=\frac{\log n}{2^{s+1}}$ and 
since $\gamma=1/2^{s+1}$, it holds while $\log s + s + 1 < \log\log n$, which holds at $s<\frac{1}{2}\log\log n$. 
%The bound given by Theorem~\ref{thm:Ress-bin-k-clique-bis} at this point is $n^{\frac{\log n}{2\log\log n}}$ which is indeed superpolynomial.
\end{proof}

\begin{lemma}
\label{lem:probcond} Let $A,B,C$ three events such that $\Pr[A],\Pr[B],\Pr[C]>0$:
\begin{enumerate}
\item[$(i)$] If $\Pr[A|\neg B]\geq \Pr[A]$ then $\Pr[A | B]\leq \Pr[A]$;
\item[$(ii)$] 
 $\Pr[A|B] \geq \Pr[A]$ and   $\Pr[A|C] \geq \Pr[A]$. Then
 $\Pr[A | B \vee C] \geq \Pr[A]$.
\end{enumerate}
\end{lemma}
\begin{proof}
For part (i) consider the following equivalences:
$$
\begin{array}{lll}
\Pr[A]&=&\Pr[A|B]\Pr[B]+\Pr[A|\neg B]\Pr[\neg B] \\
   \Pr[A]       &=&\Pr[A|B]\Pr[B]+\Pr[A|\neg B](1-\Pr[B]) \\
   \Pr[A]       &\geq& \Pr[A|B]\Pr[B]+\Pr[A](1-\Pr[B])\\
   \Pr[A]\Pr[B]       &\geq&\Pr[A|B]\Pr[B]\\
   \Pr[A]       &\geq&\Pr[A|B] \\
   \end{array}
$$
For part (ii) consider the following inequalities: 
$$
\begin{array}{lll}
\Pr[A | B \vee C] &=& \frac{\Pr[A \wedge (B \vee C)]}{\Pr[B\vee C]} \\
 &\geq & \frac{\Pr[A \wedge B]}{\Pr[B\vee C]} +\frac{\Pr[A \wedge C]}{\Pr[B\vee C]}  \\
 & = & \frac{\Pr[A \wedge B]}{\Pr[B]} \cdot \frac{\Pr[B]}{\Pr[B\vee C]}+\frac{\Pr[A \wedge C]}{\Pr[C]} \cdot \frac{\Pr[C]}{\Pr[B\vee C]} \\
 &= & \Pr[A|B] \cdot \frac{\Pr[B]}{\Pr[B\vee C]}+\Pr[A|C] \cdot \frac{\Pr[C]}{\Pr[B\vee C]} \\
 &\geq &\Pr[A] \cdot (\frac{\Pr[B]+\Pr[C]}{\Pr[B\vee C]}) \\
 & \geq &\Pr[A]
 \end{array}
$$
\end{proof}

\medskip

\begin{lemma}
\label{lem:prefect}
Let $s$ be an integer, $s \geq 1$, $\gamma=\frac{1}{2^{s+1}}$,  and $\rho$ be a $s$-random restriction. 
For all $s$-tuples $S$:
$$\Pr[\mbox{$S$ survives $\rho$}] \leq 1-\frac{\gamma^s}{2^{2s}}$$
\label{lem:survival-maximised-k-clique}
\end{lemma}
\begin{proof}

Let $T=(\ell_{i_1,j_1},\ldots, \ell_{i_s,j_s})$ be an $s$-tuple made of of disjoint literals of 
$\BinC(G)$. We say that $T$ is {\em perfect} if all literals are bits of a same block. 

We prove that $\Pr[\mbox{$T$ survives $\rho$}]\leq 1-\frac{\gamma^s}{2^{2s}}$.
The result follows observing that  $\Pr[\mbox{$T$ survives $\rho$}] \geq \Pr[\mbox{$S$ survives $\rho$}].$

Let $\gamma=\frac{1}{2^{s+1}}$. A block with $r$ distinct bits contributes a factor of
\[ \frac{{\gamma \log n \choose r}}{{\log n \choose r}} \cdot \frac{1}{2^r} \]
to the probability that the $s$-tuple \textbf{does not} survive. Expanding the left-hand part of this we obtain
\[ \frac{\gamma \log n \cdot \gamma \log n \, -1 \cdots \gamma \log n \, -r+1}{\log n \cdot \log n \, -1 \cdots \log n \, -r+1} = \gamma \frac{\log n}{\log n} \cdot \gamma \frac{\log n \, -\frac{1}{\gamma}}{\log n \, -1} \cdots \gamma \frac{\log n \, -\frac{r}{\gamma}+\frac{1}{\gamma}}{\log n \, -r+1} \]
Next, let us note that
\[  1 = \frac{\log n}{\log n} > \frac{\log n \, -\frac{1}{\gamma}}{\log n \, -1} > \cdots > \frac{\log n \, -\frac{r}{\gamma}+\frac{1}{\gamma}}{\log n \, -r+1} \]
The result now follows when we recall that the probability of surviving is maximised when the probability of not surviving is minimised.
\end{proof}
In the sequel we will use the fact that, while $r<\frac{\gamma}{2} \cdot \log n$, 
\[ \frac{{\gamma \log n \choose r}}{{\log n \choose r}} \cdot \frac{1}{2^r} \geq \frac{\gamma^r}{2^{2r}} \]
since, for such $r$, $\frac{\log n \, -\frac{r}{\gamma}+\frac{1}{\gamma}}{\log n \, -r+1} > \frac{1}{2}$.

%Using previous Lemma \ref{lem:ind}, Lemma \ref{lem:Res1-bin-k-clique} and Corollary \ref{cor:BinC}, we immediately have:
%\pagebreak

\section{The weak Pigeonhole Principle}
\label{sec:wphp}
 
For $n<m$, let $\BinPHP^m_n$ be the binary encoding of the (weak) Pigeonhole Principle as showed in The Introduction in Subsection \ref{subsec:theory}.  First notice that an analogous of 
Lemma \ref{lem:reslog}  holds for the pigeonhole principle too.
\begin{lemma}
\label{lem:php-un-bin}
Suppose there are Resolution refutations of $\pPHP^m_n$ of  size $S$. 
Then there are $\RES(\log n)$ refutations of $\BinPHP^m_n$  of size $S$.
\end{lemma}

Let $\rho$ be a partial assignment (a restriction) to the variables of $\BinPHP^m_n$. We call $\rho$ a {\em $t$-bit} restriction if $\rho$ assigns $t$ bits of each pigeon $b\in [m]$, i.e. $t$ variables $\omega_{b,i}$  for each pigeon $b$. Let $v=(i,a)$  be an assignment meaning that pigeon $i$ is assigned to hole $a$ and let
$a_1\dots a_{\log n}$ be the binary representation of $a$.  We say that a restriction $\rho$ is {\em consistent} with $v$  if for all $j\in [\log n]$, $\sigma(\omega_{i,j})$ is either $a_j$ or not assigned. We denote by $\BinPHP^m_n\!\!\!\upharpoonright_\rho$, $\BinPHP^m_n$ restricted by $\rho$. We will also consider the situation in which an $s$-bit restriction is applied to some $\BinPHP^m_n\!\!\!\upharpoonright_\rho$, creating $\BinPHP^m_n\!\!\!\upharpoonright_\tau$, where $\tau$ is an $s+t$-bit restriction.

Throughout this section, let $u=u(n,t):=(\log n) - t$. We do not use this shorthand universally, but sometimes where otherwise the notation would look cluttered. We also occasionally write $(\log n) - t$ as $\log n\, - t$ (note the extra space). 
%\red{Nicola: we need some kind of convention here.}%Call a pigeon \emph{busy} 
%in a record if it is mentioned in a literal either positively or negatively.
 
%For $n<N$, let $\phi^m_{n}(\PHP^m_N)$ be the binary encoding of the Pigeonhole Principle $\PHP^m_N$ in which, for each pigeon, only $n$ holes remain available. Note that this need not be the result of a random restriction. While $m>N$, this remains a contradiction.
\begin{lemma}
\label{lem:unwritten}
Let $\rho$ be a $t$-bit restriction for $\BinPHP^m_n$. Any decision DAG for $\BinPHP^m_n\!\!\!\upharpoonright_\rho$ must contain a record which mentions $\frac{n}{2^{t}}$ pigeons.
\end{lemma}
\begin{proof}
Let Adversary play in the following fashion. While some pigeon is not mentioned at all, let him give Prover a free choice to answer any one of its bits as true or false. Once a pigeon is mentioned once, then let Adversary choose a hole for that pigeon by choosing some assignment for the remaining unset bits (we will later need to prove this is always possible). Whenever another bit of an already mentioned pigeon is queried, then Adversary will answer consistently with the hole he has chosen for it. Only once all of a pigeon's bits are forgotten (not including those set by $\rho$), will Adversary forget the hole he assigned it.

It remains to argue that Adversary must force Prover to produce a record of width $\geq \frac{n}{2^{t+1}}$ and for this it suffices to argue that Adversary can remain consistent with $\BinPHP^m_n\!\!\!\upharpoonright_\rho$ up until the point that such a record exists. For that it is enough to show that there is always a hole available for a pigeon for which Adversary gave its only currently questioned bit as a free choice (but for which $\rho$ has already assigned some bits).

The current record is assumed to have fewer than $\frac{n}{2^{t}}$ literals and therefore must mention fewer than $\frac{n}{2^{t}}$ pigeons, each of which Adversary already assigned a hole. Each hitherto unmentioned pigeon that has just been given a free choice has $\log n \ - t$ bits which corresponds to $\frac{n}{2^t}$ holes. Since we have assigned fewer than $\frac{n}{2^t}$ pigeons  to holes, one of these must be available, and the result follows. 
\end{proof}

\noindent Let $\xi(s)$ satsify $\xi(1)=1$ and $\xi(s)=\xi(s-1)+1+s$. Note that $\xi(s)=\Theta(s^2)$.

%\begin{definition}[Property $\pPHP(s)$]
%Let $s\geq 1$. For any $s$-bit restriction $\rho$ to $\BinPHP(m,n)$, there are no $\RES(s)$ refutations of $\BinPHP(m,n)\!\!\!\upharpoonright_\rho$ of size smaller than $e^{\frac{n}{4^{\xi(s)+1} s! \log^{\xi(s)}n}}$.
%\end{definition}

\begin{definition}[Property $\pPHP(s,t)$]
Let $s,t\geq 1$. For any $t$-bit restriction $\rho$ to $\BinPHP^m_n$, there are no $\RES(s)$ refutations of $\BinPHP^m_n\!\!\!\upharpoonright_\rho$ of size smaller than $e^{\frac{n}{4^{\xi(s)+1} s! 2^t u^{\xi(s)}}}$.
\end{definition}

\begin{theorem}
\label{thm:reswphp}
Let $\rho$ be a $t$-bit restriction for $\BinPHP^m_n$. Any decision DAG for $\BinPHP^m_n\!\!\!\upharpoonright_\rho$ is of size $2^{\Omega(\frac{n}{\log n})}$ (indeed, asymptotically of size $\geq e^{\frac{n}{2^{t+2} u}}$).
\end{theorem}
\begin{proof}
Call a \emph{bottleneck} a record in the decision DAG that mentions $\frac{n}{2^{t+1}}$ pigeons. Now consider a random restriction that picks for each pigeon one bit uniformly at random and sets this to $0$ or $1$ with equal probability. The probability that a bottleneck survives (is not falsified by) the random restriction is no more than 
\[ \left( \frac{u-1}{u} + \frac{1}{2 u} \right)^{\frac{n}{2^{t+1}}} = \left( 1 - \frac{1}{2 u} \right)^{u \cdot \frac{n}{2^{t+1} u}} \leq \frac{1}{e^{\frac{n}{2^{t+2} u}}}, \]
since $e^{-x} = \lim_{m\to\infty} (1 - x/m)^m$ and indeed $e^{-x} \geq (1 - x/m)^m$ when $x,m \geq 1$.

Now suppose for contradiction that we have fewer than $e^{\frac{n}{2^{t+2} u}}$ bottlenecks in a decision DAG for $\BinPHP^m_n\!\!\!\upharpoonright_\rho$. By the union bound there is a random restriction that kills all bottlenecks and this leaves a decision DAG for some $\BinPHP^m_n\!\!\!\upharpoonright_\sigma$, where $\sigma$ is a $(t+1)$-bit restriction for $\BinPHP^m_n$.
However, we know from Lemma~\ref{lem:unwritten} that such a refutation must involve a record mentioning $\frac{n}{2^{t+1}}$ pigeons. This is now the desired contradiction.
\end{proof}
Note that the previous theorem could have been proved, like Lemma~\ref{lem:Res1-bin-k-clique}, by the size-width trade-off. However, the method of random restrictions used here could not be easily applied there, due to the randomness of $G$.
\begin{corollary}
Property $\pPHP(1,t)$ holds, for each $t<\log n$.
\end{corollary}
\noindent Note that, $\pPHP(1,t)$ yields only trivial bounds as $t$ approaches $\log n$.

Let $(\ell_{i_1,j_1},\ldots, \ell_{i_s,j_s})$ be an $s$-tuple made of disjoint literals of 
$\BinPHP^m_n\upharpoonright_\rho$. We say that a tuple is {\em anti-perfect} if all literals come from different pigeons. 

\begin{lemma}
Let $s$ be an integer, $s\geq 1$ and $\sigma$ an $s$-bit restriction over $\BinPHP^m_n \!\!\!\upharpoonright_\rho$ where $\rho$ is itself some $t$-bit restriction over $\BinPHP^m_n$. Let $T$ be an anti-perfect $s$-tuple of $\BinPHP^m_n\!\!\!\upharpoonright_\rho$. Then for all $s$-tuples S:
$$\Pr[\mbox{$T$ survives $\sigma$}] \geq \Pr[\mbox{$S$ survives $\sigma$}].$$
and so $\Pr[\mbox{$S$ survives $\sigma$}] \leq 1 - \frac{1}{(\log n\, -t)^s 2^s} = 1 - \frac{1}{u^s 2^s}$.
\label{lem:survival-maximised-PHP}
\end{lemma}
\begin{proof}
A pigeon with $r$ distinct bits contributes a factor of
\[ \frac{r }{\log n \, -t} \cdot \frac{r-1}{\log n \, -t -1} \cdots \frac{1}{\log n -t -r+1} \cdot \frac{1}{2^r}.\]
Noting that
\[ \frac{r!}{\log n \, -t \cdot \log n \, -t -1 \cdots \log n \, -t -r+1} > \frac{1}{(\log n \, -t)^r}, \]
the result now follows when we recall that the probability of surviving is maximised when the probability of not surviving is minimised.
\end{proof}

\begin{theorem}
Let $s>1$ and $s+t<\log n$. Then, $\pPHP(s-1,s+t)$ implies $\pPHP(s,t)$.
\label{thm:new-wphp-main}
\end{theorem}
\begin{proof}
We proceed by contraposition. Assume there is some $t$-bit restriction $\rho$ so that there exists a $\RES(s)$ refutation $\pi$ of $\BinPHP^m_n\!\!\!\upharpoonright_\rho$ with size less than $e^{\frac{n}{4^{\xi(s)+1} \cdot s! 2^t u^{\xi(s)}}}$.

Call a \emph{bottleneck} a record that has covering number $\geq \frac{n}{4^{\xi(s)} \cdot (s-1)! 2^t u^{\xi(s-1)}}$.  In such a record, by dividing by $s$ and $u$, it is always possible to find $r:=\frac{n}{4^{\xi(s)} s! 2^t u^{\xi(s-1)+1}}$ $s$-tuples \sloppy of literals $(\ell^1_1,\ldots,\ell^s_1),\ldots,(\ell^1_r,\ldots,\ell^s_r)$ so that each $s$-tuple is a clause in the record and no pigeon appearing in the $i$th $s$-tuple also appears in the $j$th $s$-tuple (when $i \neq j$). This important independence condition plays a key role. Now consider a random restriction that, for each pigeon, picks uniformly at random $s$ bit positions and sets these to $0$ or $1$ with equal probability. The probability that the $i$th of the $r$ $s$-tuples survives the restriction is maximised when each variable among the $s$ describes a different pigeon (by Lemma~\ref{lem:survival-maximised-PHP}) and is therefore bound above by
\[ \left( 1 - \frac{2^s-1}{2^s u^s} \right) \]
whereupon
\[ \left( 1 - \frac{2^s-1}{2^s u^s} \right)^{\frac{n}{4^{\xi(s)} s! 2^t u^{\xi(s-1)+1}}} = \left( 1 - \frac{2^s-1}{2^s u^s} \right)^{\frac{n u^s}{4^{\xi(s)}s! 2^t u^{(\xi(s-1)+1+s)}}} \]
which is $\leq 1/e^{\frac{(2^s-1) n}{4^{\xi(s)}s! \cdot 2^s 2^{t} u^{\xi(s)}}}<1/e^{\frac{n}{4^{\xi(s)+1}s! \cdot 2^{t} u^{\xi(s)}}}$.
Supposing therefore that there are fewer than $e^{\frac{n}{4^{\xi(s)+1}s! \cdot 2^{t} u^{\xi(s)}}}$ bottlenecks, one can deduce a random restriction that kills all bottlenecks. What remains after doing this is a $\RES(s)$ refutation of some $\BinPHP^m_n\!\!\!\upharpoonright_\sigma$, where $\sigma$ is a $s+t$-bit restriction, which moreover has covering number $< \frac{n}{4^{\xi(s)} \cdot (s-1)! 2^t u^{\xi(s-1)}}$. But if the remaining $\RES(s)$ refutation is of size $<e^{\frac{n}{4^{\xi(s)+1}s! \cdot 2^{t} u^{\xi(s)}}}$ then,  from Lemma~\ref{lem:covering-number}, it would give a $\RES(s-1)$ refutation of size
\[ <2^{\frac{n}{4^{\xi(s)} \cdot (s-1)! 2^t u^{\xi(s-1)}}} \cdot e^{\frac{n}{4^{\xi(s)+1}s! \cdot 2^{t} u^{\xi(s)}}} = e^{\frac{n}{4^{\xi(s)} \cdot (s-1)! 2^t u^{\xi(s-1)}} (\ln 2 + \frac{1}{4 s u^{s+1}})} \]
\[< e^{\frac{2n}{4^{\xi(s)} \cdot (s-1)! 2^t u^{\xi(s-1)}}} < e^{\frac{n}{4^{\xi(s)} \cdot (s-1)! 2^{t+1} u^{\xi(s-1)}}} <  e^{\frac{n}{4^{\xi(s)-s} \cdot (s-1)! 2^{s+t} u^{\xi(s-1)}}},\]
since $4^{s} > 2^{s-1}$, which equals $e^{\frac{n}{4^{\xi(s-1)+1} \cdot (s-1)! 2^{s+t} u^{\xi(s-1)}}}$ in contradiction to the inductive hypothesis.
%\[<2^{\frac{n}{4^{\xi(s)} (s-1)!\log^{\xi(s-1)} n}}\cdot e^{\frac{(2^s-1) n}{4^{\xi(s)} s! \cdot 2^s\log^{\xi(s)}n}}= e^{\frac{n}{4^{\xi(s)} (s-1)!\log^{\xi(s-1)} n}(\ln 2 + \frac{(2^s-1) n}{s \cdot 2^s\log^{s+1}n})}\]
%\[< e^{\frac{2n}{4^{\xi(s)} (s-1)!\log^{\xi(s-1)} n}}< e^{\frac{n}{4^{\xi(s)-1} (s-1)!\log^{\xi(s-1)} n}} < e^{\frac{n/2^s}{4^{\xi(s)-1-s} (s-1)!\log^{\xi(s-1)} (n/2^s)}}\]
%$= e^{\frac{n/2^s}{4^{\xi(s-1)} (s-1)!\log^{\xi(s-1)} (n/2^s)}}$,
%in contradiction to the inductive hypothesis.
\end{proof}

\begin{theorem}
\label{thm:wphp}
Fix $\lambda,\mu>0$. Any refutation of $\BinPHP^m_n$ in $\RES(\sqrt{2}\log^{\frac{1}{2+\lambda}}n)$ is of size $2^{\Omega(n^{1-\mu})}$.
\end{theorem}
\begin{proof}
First, let us claim that $\pPHP(\sqrt{2}\log^{\frac{1}{2+\lambda}}n,0)$ holds (and this would hold also at $\lambda=0$). Applying Theorem~\ref{thm:new-wphp-main} gives $\ell$ such that $\frac{\ell(\ell+1)}{2}<\log n$. Noting $\frac{\ell^2}{2}<\frac{\ell(\ell+1)}{2}$, the claim follows.

Now let us look at the bound we obtain by plugging in to $e^{\frac{n}{4^{\xi(s)+1} \cdot s! 2^t u^{\xi(s)}}}$ at $s=\sqrt{2} \log^{\frac{1}{2+\lambda}}n$ and $t=0$. We recall $\xi(s)=\Theta(s^2)$. It follows,  since $\lambda>0$, that each of $4^{\xi(s)+1}$, $s!$ and $\log^{\xi(s)}n$ is $o(n^{\mu})$. The result follows.
\end{proof}

\subsection{The treelike case}
\label{subsec:treephp}
Concerning the Pigeonhole Principle,  we can prove that the relationship between $\PHP$ and $\BinPHP$ is different for treelike Resolution from general Resolution. In particular, for very weak Pigeonhole Principles, we know the binary encoding is harder to refute in general Resolution; whereas for treelike Resolution it is the unary encoding which is the harder.

\begin{theorem}
The treelike Resolution complexity of $\BinPHP^m_n$ is $2^{\Theta(n)}$.
\label{thm:phpubtreelike}
\end{theorem}
\begin{proof}
For the lower bound, one can follow the proof of Lemma~\ref{lem:unwritten} with $t=0$ and finds $n$ free choices on each branch of the tree. Following the method of Riis \cite{SorenGap}, we uncover a subtree of the decision tree of size $2^n$.

For an upper bound of $2^{2n}$ we pursue the following strategy. First we choose some $n+1$ pigeons to question. We then question all of them on their first bit and separate these into two sets $T_1$ and $F_1$ according to whether this was answered true or false. If $n$ is a power of $2$, choose the larger of these two sets (if they are the same size then choose either). If $n$ is not a power of two, the matter is mildly complicated, and one must look at how many holes are available with the first bit set to $1$, say $h^1_1$; versus $0$, say $h^0_1$. At least one of $|T_1|>h^1_1$ or $|F_1|>h^0_1$ must hold and one can choose between $T_1$ and $F_1$ correspondingly.
Now question the second bit, producing two sets $T_2$ and $F_2$, and iterate this argument. We will reach a contradiction in $\log n$ iteration sinxe we always choose a set of masimal size. The depth of our tree is bound above by $n+\frac{n}{2} + \frac{n}{4} + \cdots < 2n$ and the result follows. 
\end{proof}

\section{Contrasting unary and binary encodings}

\subsection{Binary encodings of principles involving total comparison}

We will now argue that the proof complexity in Resolution of principles involving total comparison will not increase significantly (by more than a polynomial factor) when shifting from the unary encoding to the binary encoding. \emph{Total comparison} is here indicated by the axioms $v_{i,j} \oplus v_{j,i}$, where $\oplus$ indicates XOR, for each $i \neq j$. It follows that it does not make sense to consider the binary encoding of such principles in the search for strong lower bounds. Examples of natural principles involving total comparison include the totally ordered variant of the Ordering Principle (known to be polynomially refutable in Resolution \cite{DBLP:journals/cc/BonetG01}) as well as all of its unary relativisations (which can be exponentially hard for any Res$(s)$ \cite{DantchevR03}).

Let $\TCPrin$ be some $\Pi_2$ first-order principle involving relations of arity no more than $2$. Let $n\in \mathbb N$ and discover $\TCPrin(n)$ with variables $v_{i,j}$, for $i,j \in [n]$, of arity 2, including axioms of total comparison: $v_{i,j} \oplus v_{j,i}$, for each $i \neq j$. There may additionally be unary variables, of the form $u_i$, for $i \in [n]$, but no further variables of other arity. Let $\UnTCPrin(n)$ have axioms $v_{i,1} \vee \ldots \vee v_{i,n}$, for each $i \in [n]$ (for the Ordering Principle this would most naturally correspond to the variant stating a finite total order has a maximal element). To make our translation to the binary encoding, we tacitly assume $n$ is a power or $2$. When this is not the case, we need clauses forbidding certain evaluations, and we defer this treatment to Section~\ref{sec:gentheo}. Let $\BinTCPrin(n)$ have corresponding variables $\omega_{i,\ell}$ for $i \in [n], \ell \in [\log n]$, where $v_{i,j}$ from the unary encoding semantically corresponds to the  conjunction $(\omega^{a_1}_{i,1}\wedge \ldots \wedge \omega^{a_{\log n}}_{i,{\log n}})$,
where
$$
\omega^{a_p}_{i,p}=
\left\{
\begin{array}{ll}
\omega_{i,p} & \mbox{ if $a_p=1$} \\
\overline \omega_{i,p} & \mbox{ if $a_p=0$}
\end{array}
\right.
$$ 
with $a_1\cdots a_{\log n}$ being the binary representation of $j$. The unary variables stay as they are. From this, the axioms of $\BinTCPrin(n)$, including total comparison, can be canonically calculated from the corresponding axioms of $\UnTCPrin(n)$ as explained in Section \ref{sec:gentheo} in 
Defintion \ref{def:binC}. Note that the large disjunctive clauses of $\UnTCPrin(n)$, that encode the existence of the witness, disappear completely in $\BinTCPrin(n)$.
\begin{lemma}
\label{lem:totoord}
Suppose there is a Resolution refutation of $\UnTCPrin(n)$ of size $S(n)$. Then there is a Resolution refutation of $\BinTCPrin(n)$ of size at most $n^2\cdot S(n)$.
\end{lemma}
\begin{proof}
Take a decision DAG $\pi$ for $\UnTCPrin(n)$ and consider the point at which some variable $v_{i,j}$ is questioned. Each node in $\pi$ will be expanded to a small tree in $\pi'$, which will be a decision DAG for $\BinTCPrin(n)$. The question ``$v_{i,j}?$'' in $\pi$ will become a sequence of $2\log n$ questions on variables $\omega_{i,1}, \ldots, \omega_{i,\log n},\omega_{j,1}, \ldots,\omega_{j,\log n}$, giving rise to a small tree of size $2^{2\log n}=n^2$ questions in $\pi'$. Owing to total comparison, many of the branches of this mini-tree must end in contradiction. Indeed, many of their leaves would imply the impossible $\neg v_{i,j} \wedge \neg v_{j,i}$, while precisely one would imply the impossible $v_{i,j} \wedge v_{j,i}$ (see Figure~\ref{fig:unary2binary} for an example). Those that don't will always have a sub-branch labelled by $(\omega^{a_1}_{i,1}\wedge \ldots \wedge \omega^{a_{\log n}}_{i,{\log n}})$,
where
$$
\omega^{a_p}_{i,p}=
\left\{
\begin{array}{ll}
\omega_{i,p} & \mbox{ if $a_p=1$} \\
\overline  \omega_{i,p} & \mbox{ if $a_p=0$}
\end{array}
\right.
$$
with $a_1\cdots a_{\log n}$ being the binary representation of $j$; \textbf{or} $(\omega^{b_1}_{j,1}\wedge \ldots \wedge \omega^{b_{\log n}}_{j,{\log n}})$,
where
$$
\omega^{b_p}_{j,p}=
\left\{
\begin{array}{ll}
\omega_{j,p} & \mbox{ if $b_p=1$} \\
\overline \omega_{j,p} & \mbox{ if $b_p=0$}
\end{array}
\right.
$$
with $b_1\cdots b_{\log n}$ being the binary representation of $i$. By forgetting information along these branches and unifying branches with the same labels of their sub-branches, we are left with precisely these two outcomes, corresponding to ``$v_{i,j}$'' and ``$\neg v_{i,j}$'' (which is ``$v_{j,i}$''). Thus, $\pi$ gives rise to $\pi'$ of size $n^2\cdot S(n)$ and the result follows.
\end{proof}

\begin{figure}[!htbp]
\centering
\label{fig:unary2binary}
\mbox{\Huge
$
\resizebox{!}{0.06cm}{
\xymatrix{
& & & & & &  & & & & & & & & & \omega_{2,1} \ar[dllllllll] \ar[drrrrrrrr] \\ 
& & & & & & & \omega_{2,2} \ar[dllll] \ar[drrrr] & & & & & &  & & & & & & & & & & \omega_{2,2} \ar[dllll] \ar[drrrr] & & & & & & & \\
& & & \omega_{3,1} \ar[dll] \ar[drr] & & & & & & & & \omega_{3,1} \ar[dll] \ar[drr] & & & & & & & & \omega_{3,1} \ar[dll] \ar[drr]  & & & & & & & & \omega_{3,1} \ar[dll] \ar[drr]  & & & \\
& \omega_{3,2} \ar[dl] \ar[dr] & & & & \omega_{3,2} \ar[dl] \ar[dr] & & & & \omega_{3,2} \ar[dl] \ar[dr] & & & & \omega_{3,2} \ar[dl] \ar[dr] & & & & \omega_{3,2} \ar[dl] \ar[dr] & & & & \omega_{3,2} \ar[dl] \ar[dr] & & & & \omega_{3,2} \ar[dl] \ar[dr] & & & & \omega_{3,2} \ar[dl] \ar[dr] & \\ 
\# & & \# & & B & & \# & & \# & & \# & & B & & \# & & \# & & \# & & B & & \# & & A & & A & & \# & & A \\
}
}
$
}
\caption{
Example converting the question $v_{2,3}?$ from a Resolution refutation of $\UnTCPrin(n)$ to a small tree in a refutation of $\BinTCPrin(n)$. The variables $\omega_{2,1}, \omega_{2,2}, \omega_{3,1}, \omega_{3,2}$ are questioned in order. The left-hand and right-hand branches correspond to false and true, respectively. Note that $2$ and $3$ are $10$ and $11$ in binary, respectively. Thus, $v_{2,3}$ is equivalent to $\omega_{2,1} \wedge \omega_{2,2}$ (labelled $A$ at the leaves) and $v_{3,2}$ is equivalent to $\omega_{3,1} \wedge \overline \omega_{3,2}$ (labelled $B$ at the leaves). The remaining leaves contradict the total comparison clauses (including one that would be labelled both $A$ and $B$).
}
\end{figure}
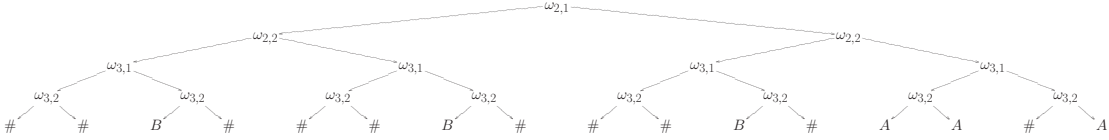

\subsection{Binary encodings of principles versus their Unary functional encodings}
\label{subsec:func}
Recall the unary functional encoding of a combinatorial principle $\mathsf{C}$, denoted $\mathsf{Un}$-$\mathsf{Fun}$-$\mathsf{C}(n)$, replaces the big clauses from $\mathsf{Un}$-$\mathsf{C}(n)$, of the form $v_{i,1} \vee \ldots \vee v_{i,n}$, with $v_{i,1} + \ldots + v_{i,n} = 1$, where addition is made on the natural numbers. This is equivalent to augmenting the axioms $\neg v_{i,j} \vee \neg v_{i,k}$, for $j \neq k \in [n]$.
\begin{lemma}
Suppose there is a Resolution refutation of $\mathsf{Bin}$-$\mathsf{C}(n)$ of size $S(n)$. Then there is a Resolution refutation of $\mathsf{Un}$-$\mathsf{Fun}$-$\mathsf{C}(n)$ of size at most $n^2\cdot S(n)$.
\end{lemma}
\begin{proof}
Take a decision DAG $\pi'$ for $\mathsf{Bin}$-$\mathsf{C}(n)$, where \mbox{w.l.o.g.} $n$ is even, and consider the point at which some variable $\nu'_{i,j}$ is questioned. Each node in $\pi'$ will be expanded to a small tree in $\pi$, which will be a decision DAG for $\mathsf{Un}$-$\mathsf{Fun}$-$\mathsf{C}(n)$. The question ``$\nu'_{i,j}?$'' in $\pi$ will become a sequence of questions $v_{i,1},\ldots,v_{i,n}$ where we stop the small tree when one of these is answered true, which must eventually happen. Suppose $v_{i,k}$ is true. If the $j$th bit of $k$ is $1$ we ask now all $v_{i,b_1},\ldots,v_{i,b_{\frac{n}{2}}}$, where $b_1,\ldots,b_{\frac{n}{2}}$ are precisely the numbers in $[n]$ whose $j$th bit is $0$. All of these must be false. Likewise, if the $j$th bit of $k$ is $0$ we ask all $v_{i,b_1},\ldots,v_{i,b_{\frac{n}{2}}}$, where $b_1,\ldots,b_{\frac{n}{2}}$ are precisely the numbers whose $j$th bit is $1$. All of these must be false. We now unify the branches on these two possibilities, forgetting any intermediate information. (To give an example, suppose $j=2$. Then the two outcomes are $\neg v_{i,1} \wedge \neg v_{i,3} \wedge  \ldots \wedge \neg v_{i,n-1}$ and $\neg v_{i,2} \wedge \neg v_{i,4} \wedge  \ldots \wedge \neg v_{i,n}$.) Thus, $\pi'$ gives rise to $\pi$ of size $n^2\cdot S(n)$ and the result follows.
\end{proof}

\subsection{The Ordering Principle in binary}

Recall the Ordering Principle specified in $\Pi_2$ first-order logic
\[\forall x,y,z \exists w \ \neg R(x,x) \wedge (R(x,y) \wedge R(y,z) \rightarrow R(x,z)) \wedge R(x,w)\]
with propositional translation to the binary encoding of witnesses, $\mathsf{Bin}$-$\mathsf{OP}_n$, as follows.
$$
\begin{array}{ll}
		\overline \nu_{x,x} & x \in [n] \\
		\overline \nu_{x,y} \vee \overline  \nu _{y,z} \vee \nu_{x,z} & x,y,z \in [n] \\
		\bigvee_{i\in [\log n]} \omega^{1-a_i}_{x,i} \vee \nu_{x,a} & x,a \in [n] \\
\end{array}
$$

%\[
%\begin{array}{cl}
%\neg r_{xx} & \mbox{for all $x \in [n]$} \\
%\neg R(x,y) \vee \neg R(y,z) \vee R(x,z) & \mbox{for all $x,y,z \in [n]$} \\
%W^{1-a_1}(x,1)\vee  \ldots \vee W^{1-a_{\log n}}(x,\log n) \vee R(x,a) & \mbox{for all $x,a \in [n]$}
%\end{array}
%\]
where 
$$\omega^{a_j}_{i,j}=
\left\{
\begin{array}{ll}
\omega_{i,j}& \mbox{ if $a_j=1$} \\
\overline \omega_{i,j} & \mbox{ if $a_j=0$}
\end{array}
\right.
$$
and $a_1\ldots a_{\log n}$ is the binary representation of $a$.
\begin{lemma}
$\mathsf{Bin}$-$\mathsf{OP}_n$ has refutations in Resolution of polynomial size.
\end{lemma}
\begin{proof}
We follow the well-known proof for the unary version of the Ordering Principle, from \cite{DBLP:journals/acta/Stalmarck96}. Consider the domain to be $[n]=\{1,\ldots,n\}$. At the $i$th stage of the decision DAG we will find a maximal element, ordered by $R$, among $[i]=\{1,\ldots,i\}$. That is, we will have a record of the \emph{special} form 
%\[R(1,j) \wedge \ldots \wedge R(j-1,j) \wedge R(j+1,j) \wedge \ldots \wedge R(i,j)\]
\[\overline \nu_{j,1} \wedge \ldots \wedge \overline  \nu_{j,j-1} \wedge \overline  \nu_{j,j+1} \wedge \ldots \wedge \overline  \nu_{j,i}\]
for some $j \in [i]$. The base case $i=1$ is trivial. Let us explain the inductive step. From the displayed record above we ask the question $ \nu_{j,i+1}?$ If $ \nu_{j,i+1}$ is true, then ask the sequence of questions $ \nu_{i+1,1},\ldots, \nu_{i+1,i}$, all of which must be false by transitivity. Now, by forgetting information, we uncover a new record of the special form. Suppose now $ \nu_{j,i+1}$ is false. Then we equally have a new record again in the special form. Let us consider the size of our decision tree so far. There are $n^2$ nodes corresponding to special records and navigating between special records involves a path of length $n$, so we have a DAG of size $n^3$. Finally, at $i=n$, we have a record of the form
\[\overline \nu_{j,1} \wedge \ldots \wedge \overline \nu_{j,j-1} \wedge \overline  \nu_{j,j+1} \wedge \ldots \wedge \overline  \nu_{j,n}.\]
Now we expand a tree questioning the sequence $w_{j,1},\ldots,w_{j,\log n}$, and discover each leaf labels a contradiction of the clauses of the final type. We have now added $n \cdot 2^{\log n}$ nodes, so our final DAG is of size at most $n^3+n^2$. 
\end{proof}

%\TODO{We have to write  down precisely the statement of the Theorem stating the upper bound for binary $\LOP$ with antisimmetry and the $\LOP'$ with  used by Alekhnocivh to separate   regular Resolution from  Resolution}

\begin{theorem} 
\label{thm:mainlop}
$\bOP$ has poly size resolution refutations in $\RES(1)$.
\end{theorem}
$\bRLOP$ is a family of contradictions based on a variant of the Ordering Principle, which is important as it exponentially separates read-once Resolution from Resolution (see \cite{AJPU}).
\begin{corollary} 
\label{cor:ale} 
$\bRLOP$ has poly size resolution refutations in $\RES(\frac{1}{2} \log \log n)$. 
\end{corollary}

\section{Binary versus unary encodings in general}
\label{sec:gentheo}

Let $\CC$ be some combinatorial principle  expressible  as a first-order $\Pi_2$-formula $F$ of the form
$\forall \vec x \exists \vec w \varphi(\vec x,\vec w)$ where $\varphi(\vec x,\vec w)$ is a quantifier-free formula built on a family of 
relations $\vec R$. Following Riis \cite{SorenGap} we restrict to the class of such formulas having no finite model.  
%Notice that
% we can restrict to the case in which only one witness variable $w$ is given, since we can encode all the witness $\vec w$, which all depend on $\vec{x}$, as one single variable $w$.  So we can assume $F$ to be fo the form $\forall \vec x \exists w \varphi(\vec x,w)$.

Let $\UC$  be the standard unary (see Riis in \cite{SorenGap}) CNF propositional encoding of $F$. 
For each set of first-order variables $\vec a:=\{x_1, \ldots, x_k\}$ of (first order) variables, we consider the  propositional variables $v_{x_{i_1},x_{i_2}, \dots ,x_{i_k}}$ (which we abbreviate as  $v_{\vec a}$)  whose semantics are to capture at once the value of variables in $\vec a$ if they appear in some relation in $\varphi$. For easiness of description we restrict to the case where $F$ is of the form $\forall \vec x \exists w \varphi(\vec x,w)$, \mbox{i.e.} ${\vec w}$ is a single variable $w$. Hence the propositional variables of $\UC$ are of the type $v_{\vec a}$ for $\vec a\subseteq \vec x$ (type 1 variables) and/or of the type $v_{\vec xw}$ for $w \in {\vec w}$ (type 2 variables) and which we denote by simply $v_{w}$, 
since each existential variable in $F$ depends always on all universal variables.% In the following definition we assume that ${\vec w}$ is a single variable $w$, as it was in all $\Pi_2$ principles appearing in this paper. 
Notice that we consider the case of $F= \forall \vec x \exists w \varphi(\vec x,w)$, since the  generalisation to higher arity is clear as each witness $w \in {\vec w}$ may be treated individually.

\begin{definition}(Canonical form of $\BC$)
\label{def:binC}
Let $\CC$ be a  combinatorial principle  expressible as a first-order formula $\forall \vec x \exists w \varphi(\vec x,w)$ with no finite models. Let $\UC$ be its unary propositional encoding. Let $2^{r-1}<n\leq 2^r \in \mathbb N$ ($r=\lceil \log n\rceil$). The binary encoding $\BC$ of $C$ is defined as follows:

The {\em \bf variables} of $\BC$ are defined from variables of $\UC$ as follows: 
\begin{enumerate}
\item For each variable of type 1 $v_{\vec a}$,  for $\vec a \subseteq \vec x$,  we use a variable $\nu_{\vec x}$,  for  $\vec a \subseteq \vec x$, and
\item  For each variable of type 2 $v_w$, we have $r$ variables $\omega_1,\ldots \omega_r$, where we use the convention that if 
$z_1\ldots z_{r}$ is the binary representation of $w$, then 
$$ 
\omega^{z_j}_{j}=\left\{
   		\begin{array}{ll}
            		\omega_{j} &  z_j=1 \\
            		\overline \omega_{j} & z_j=0
    		\end{array} \right. 
$$ 
so that $v_{w}$ can be represented using binary variables by the clause $(\omega^{1-z_1}_{1} \vee \ldots \vee \omega^{1-z_r}_{r})$
\end{enumerate}

The  {\em \bf clauses} of $\BC$ are defined form the clauses of $\UC$ as follows:
\begin{enumerate}
\item If $C \in \UC$ contains only variables of type 1, $v_{\vec b_1},\dots, v_{\vec b_k}$, hence $C$ is mapped as follows
$$
\begin{array}{lll}
C:= \bigvee_{j=1}^{k_1} v_{\vec b_j} \vee \bigvee_{j=1}^{k_2} \overline v_{\vec c_j} &\mapsto &\bigvee_{j=1}^{k_1} \nu_{\vec b_j} \vee \bigvee_{j=1}^{k_2} \overline \nu_{\vec c_j}
\end{array}
$$

%then  is mapped into the  clause on variables $\nu_{\vec b_1},\dots,  \nu_{\vec b_k}$ $$\bigvee_{j=1}^{k_1} \nu_{\vec b_j} \vee \bigvee_{j=1}^{k_2} \overline \nu_{\vec c_j} $$

\item If $C \in \UC$ contains type 1  and type 2 variables, it is mapped as follows:
$$
\begin{array}{lll}
C:= v_{w} \vee \bigvee_{j=1}^{k_1} v_{\vec c_j} \vee \bigvee_{l=1}^{k_2} \overline v_{\vec d_j} & \mapsto & \left( 
\bigvee_{i \in [r]} \omega^{1-z_i}_{i}\right ) \vee \bigvee_{j=1}^{k_1} \nu_{\vec c_j} \vee \bigvee_{l=1}^{k_2} \overline \nu_{\vec d_j}\\
C:= \overline v_{w} \vee \bigvee_{j=1}^{k_1} v_{\vec c_j} \vee \bigvee_{l=1}^{k_2} \overline v_{\vec d_j} & \mapsto & \left( 
\bigvee_{i \in [r]} \omega^{z_i}_{i}\right ) \vee \bigvee_{j=1}^{k_1} \nu_{\vec c_j} \vee \bigvee_{l=1}^{k_2} \overline \nu_{\vec d_j}\\
\end{array}
$$
where $\vec c_j,\vec d_l\subseteq \vec x$ and where $z_1,\ldots,z_r$ is the binary representation of $w$.

\item If $n\neq 2^r$, then, for each $n<a\leq 2^r$ we need clauses
\[ \omega_{1}^{1-a_1} \vee \ldots \vee \omega_{r}^{1-a_r} \]
where $a_1,\ldots,a_r$ is the binary representation of $a$.
%$$
%\omega^{a_j}_{j}=\left\{
%   		\begin{array}{ll}
%            		\omega_{j} &  a_j=1 \\
%            		\overline \omega_{j} & a_j=0
%    		\end{array} \right. 
%$$ 
\end{enumerate}
\end{definition}

Getting short proofs for the binary version $\BC$ in $\RES(\log n)$ form short $\RES(1)$ proofs of the unary version $\UC$ is possible also in the general case.

\begin{lemma}
\label{lem:resloggen}
Let $\CC$ be a  combinatorial principle  expressible as a first-order formula $\forall \vec x \exists \vec w \varphi(\vec x,\vec w)$ with no finite models. 
Let $\UC$ and $\BC$ be respectively the unary and binary propositional encoding. Let $n \in \mathbb N$. 
If there is a size $S$ refutation for $\UC$ in $\RES(1)$, then there is a size $S$ refutation for $\BC$ in $\RES(\log n)$
\end{lemma}
\begin{proof} (Sketch)
Where the decision DAG for $\UC$ questions some variable $v_{\vec{a},b}$, the decision branching $\log n$-program questions instead 
$(\omega^{1-z_1}_{\vec{a},1}\vee \ldots \vee \omega^{1-z_{\log n}}_{\vec{a},{\log n}})$ where the out-edge marked true in the former becomes false in the latter, and vice versa. What results is indeed a decision branching $\log n$-program for $\BC$, and the result follows.
\end{proof}

As one can easily notice reading Subsection \ref{subsec:theory}, the binary version $\BinPHP$ of the Pigeonhole principle we displayed there, is different from the one we would 
get applying the canonical transformation of  Definition \ref{sec:gentheo}. But we can easily and efficiently move between these versions in Resolution.  We leave the proof to the reader.

\begin{lemma}
\label{lem:eqphp}
The two versions of the binary Pigeonhole Principle ($\BinPHP$ and the one arising from Definition \ref{sec:gentheo} to $\pPHP$) 
are linearly equivalent in Resolution.
\end{lemma}

\bibliographystyle{acm}
\bibliography{Cutwidth-1}%ProofComplexity}

\begin{thebibliography}{10}

\bibitem{DBLP:journals/cc/Alekhnovich11}
{\sc Alekhnovich, M.}
\newblock Lower bounds for k-dnf resolution on random 3-cnfs.
\newblock {\em Computational Complexity 20}, 4 (2011), 597--614.

\bibitem{AJPU}
{\sc Alekhnovich, M., Johannsen, J., Pitassi, T., and Urquhart, A.}
\newblock An exponential separation between regular and general resolution.
\newblock {\em Theory of Computing 3}, 1 (2007), 81--102.

\bibitem{DBLP:journals/tcs/Atserias03}
{\sc Atserias, A.}
\newblock Improved bounds on the weak pigeonhole principle and infinitely many
  primes from weaker axioms.
\newblock {\em Theor. Comput. Sci. 295\/} (2003), 27--39.

\bibitem{DBLP:conf/stoc/AtseriasBRLNR18}
{\sc Atserias, A., Bonacina, I., de~Rezende, S.~F., Lauria, M.,
  Nordstr{\"{o}}m, J., and Razborov, A.~A.}
\newblock Clique is hard on average for regular resolution.
\newblock In {\em Proceedings of the 50th Annual {ACM} {SIGACT} Symposium on
  Theory of Computing, {STOC} 2018, Los Angeles, CA, USA, June 25-29, 2018\/}
  (2018), I.~Diakonikolas, D.~Kempe, and M.~Henzinger, Eds., {ACM},
  pp.~866--877.

\bibitem{DBLP:journals/iandc/AtseriasBE02}
{\sc Atserias, A., Bonet, M.~L., and Esteban, J.~L.}
\newblock Lower bounds for the weak pigeonhole principle and random formulas
  beyond resolution.
\newblock {\em Inf. Comput. 176}, 2 (2002), 136--152.

\bibitem{DBLP:conf/coco/BeameIS01}
{\sc Beame, P., Impagliazzo, R., and Sabharwal, A.}
\newblock Resolution complexity of independent sets in random graphs.
\newblock In {\em Proceedings of the 16th Annual {IEEE} Conference on
  Computational Complexity, Chicago, Illinois, USA, June 18-21, 2001\/} (2001),
  {IEEE} Computer Society, pp.~52--68.

\bibitem{DBLP:conf/focs/BeameP96}
{\sc Beame, P., and Pitassi, T.}
\newblock Simplified and improved resolution lower bounds.
\newblock In {\em 37th Annual Symposium on Foundations of Computer Science,
  {FOCS} '96, Burlington, Vermont, USA, 14-16 October, 1996\/} (1996), {IEEE}
  Computer Society, pp.~274--282.

\bibitem{Ben-sasson99shortproofs}
{\sc Ben-sasson, E., and Wigderson, A.}
\newblock Short proofs are narrow - resolution made simple.
\newblock In {\em Journal of the ACM\/} (1999), pp.~517--526.

\bibitem{BeyersdorffGL10}
{\sc Beyersdorff, O., Galesi, N., and Lauria, M.}
\newblock A lower bound for the pigeonhole principle in tree-like resolution by
  asymmetric prover-delayer games.
\newblock {\em Inf. Process. Lett. 110}, 23 (2010), 1074--1077.

\bibitem{Beyersdorff:2013:TOCL}
{\sc Beyersdorff, O., Galesi, N., and Lauria, M.}
\newblock Parameterized complexity of dpll search procedures.
\newblock {\em ACM Trans. Comput. Logic 14}, 3 (Aug. 2013), 20:1--20:21.

\bibitem{DBLP:journals/toct/BeyersdorffGLR12}
{\sc Beyersdorff, O., Galesi, N., Lauria, M., and Razborov, A.~A.}
\newblock Parameterized bounded-depth frege is not optimal.
\newblock {\em {TOCT} 4}, 3 (2012), 7:1--7:16.

\bibitem{DBLP:journals/jacm/BonacinaG15}
{\sc Bonacina, I., and Galesi, N.}
\newblock A framework for space complexity in algebraic proof systems.
\newblock {\em J. {ACM} 62}, 3 (2015), 23:1--23:20.

\bibitem{DBLP:journals/siamcomp/BonacinaGT16}
{\sc Bonacina, I., Galesi, N., and Thapen, N.}
\newblock Total space in resolution.
\newblock {\em {SIAM} J. Comput. 45}, 5 (2016), 1894--1909.

\bibitem{DBLP:journals/cc/BonetG01}
{\sc Bonet, M.~L., and Galesi, N.}
\newblock Optimality of size-width tradeoffs for resolution.
\newblock {\em Computational Complexity 10}, 4 (2001), 261--276.

\bibitem{BussP97}
{\sc Buss, S.~R., and Pitassi, T.}
\newblock Resolution and the weak pigeonhole principle.
\newblock In {\em Computer Science Logic, 11th International Workshop, {CSL}
  '97, Annual Conference of the EACSL, Aarhus, Denmark, August 23-29, 1997,
  Selected Papers\/} (1997), pp.~149--156.

\bibitem{DantchevR01}
{\sc Dantchev, S.~S., and Riis, S.}
\newblock Tree resolution proofs of the weak pigeon-hole principle.
\newblock In {\em Proceedings of the 16th Annual {IEEE} Conference on
  Computational Complexity, Chicago, Illinois, USA, June 18-21, 2001\/} (2001),
  pp.~69--75.

\bibitem{DantchevR03}
{\sc Dantchev, S.~S., and Riis, S.}
\newblock On relativisation and complexity gap.
\newblock In {\em Computer Science Logic, 17th International Workshop, {CSL}
  2003, 12th Annual Conference of the EACSL, and 8th Kurt G{\"{o}}del
  Colloquium, {KGC} 2003, Vienna, Austria, August 25-30, 2003, Proceedings\/}
  (2003), M.~Baaz and J.~A. Makowsky, Eds., vol.~2803 of {\em Lecture Notes in
  Computer Science}, Springer, pp.~142--154.

\bibitem{EGM}
{\sc Esteban, J.~L., Galesi, N., and Messner, J.}
\newblock On the complexity of resolution with bounded conjunctions.
\newblock {\em Theor. Comput. Sci. 321}, 2-3 (2004), 347--370.

\bibitem{DBLP:journals/siamcomp/FilmusLNRT15}
{\sc Filmus, Y., Lauria, M., Nordstr{\"{o}}m, J., Ron{-}Zewi, N., and Thapen,
  N.}
\newblock Space complexity in polynomial calculus.
\newblock {\em {SIAM} J. Comput. 44}, 4 (2015), 1119--1153.

\bibitem{Haken}
{\sc Haken, A.}
\newblock The intractability of resolution.
\newblock {\em Theor. Comput. Sci. 39\/} (1985), 297--308.

\bibitem{DBLP:conf/focs/HrubesP17}
{\sc Hrubes, P., and Pudl{\'{a}}k, P.}
\newblock Random formulas, monotone circuits, and interpolation.
\newblock In {\em 58th {IEEE} Annual Symposium on Foundations of Computer
  Science, {FOCS} 2017, Berkeley, CA, USA, October 15-17, 2017\/} (2017),
  C.~Umans, Ed., {IEEE} Computer Society, pp.~121--131.

\bibitem{krabook95}
{\sc Kraj{\'\i}{\v c}ek, J.}
\newblock {\em Bounded arithmetic, propositional logic and complexity theory}.
\newblock Cambridge University Press, 1995.

\bibitem{DBLP:journals/acta/Krishnamurthy85}
{\sc Krishnamurthy, B.}
\newblock Short proofs for tricky formulas.
\newblock {\em Acta Inf. 22}, 3 (1985), 253--275.

\bibitem{DBLP:journals/combinatorica/LauriaPRT17}
{\sc Lauria, M., Pudl{\'{a}}k, P., R{\"{o}}dl, V., and Thapen, N.}
\newblock The complexity of proving that a graph is ramsey.
\newblock {\em Combinatorica 37}, 2 (2017), 253--268.

\bibitem{DBLP:journals/jcss/MacielPW02}
{\sc Maciel, A., Pitassi, T., and Woods, A.~R.}
\newblock A new proof of the weak pigeonhole principle.
\newblock {\em J. Comput. Syst. Sci. 64}, 4 (2002), 843--872.

\bibitem{proofs_as_games}
{\sc Pudl\'{a}k, P.}
\newblock Proofs as games.
\newblock {\em American Mathematical Monthly\/} (June-July 2000), 541--550.

\bibitem{DBLP:journals/jacm/Raz04}
{\sc Raz, R.}
\newblock Resolution lower bounds for the weak pigeonhole principle.
\newblock {\em J. {ACM} 51}, 2 (2004), 115--138.

\bibitem{10.1007/3-540-46011-X_8}
{\sc Razborov, A.~A.}
\newblock Proof complexity of pigeonhole principles.
\newblock In {\em Developments in Language Theory\/} (Berlin, Heidelberg,
  2002), W.~Kuich, G.~Rozenberg, and A.~Salomaa, Eds., Springer Berlin
  Heidelberg, pp.~100--116.

\bibitem{DBLP:journals/tcs/Razborov03}
{\sc Razborov, A.~A.}
\newblock Resolution lower bounds for the weak functional pigeonhole principle.
\newblock {\em Theor. Comput. Sci. 1}, 303 (2003), 233--243.

\bibitem{SorenGap}
{\sc Riis, S.}
\newblock A complexity gap for tree resolution.
\newblock {\em Computational Complexity 10}, 3 (2001), 179--209.

\bibitem{DBLP:journals/siamcomp/SegerlindBI04}
{\sc Segerlind, N., Buss, S.~R., and Impagliazzo, R.}
\newblock A switching lemma for small restrictions and lower bounds for k-dnf
  resolution.
\newblock {\em {SIAM} J. Comput. 33}, 5 (2004), 1171--1200.

\bibitem{DBLP:journals/acta/Stalmarck96}
{\sc St{\aa}lmarck, G.}
\newblock Short resolution proofs for a sequence of tricky formulas.
\newblock {\em Acta Inf. 33}, 3 (1996), 277--280.

\end{thebibliography}

\end{document}